\documentclass[11pt,a4paper]{article}
\usepackage{amssymb,amsmath, amsfonts}
\usepackage{graphicx,graphics}
\usepackage{mathtools}
\usepackage{bbold}
\usepackage[english]{babel}
\usepackage[utf8]{inputenc}
\usepackage{epsfig,url}
\usepackage{bbm,theorem}
\usepackage{a4wide}
\usepackage{color}
\usepackage{enumerate}
\usepackage{calrsfs}
\usepackage{bookmark}
\usepackage{amsmath}
\usepackage{amsfonts}
\usepackage{amssymb}
\usepackage{verbatim}
\usepackage{graphicx}%
\providecommand{\U}[1]{\protect\rule{.1in}{.1in}}
\DeclareMathAlphabet{\pazocal}{OMS}{zplm}{m}{n}
\newtheorem{theorem}{Theorem}[section]

\newtheorem{lemma}[theorem]{Lemma}
\newtheorem{proposition}[theorem]{Proposition}

{\theorembodyfont{\upshape}
\newtheorem{remark}[theorem]{Remark}
 }
 \numberwithin{equation}{section}
\numberwithin{theorem}{section}
\newcommand{\qed}{\hfill$\Box$}
\newenvironment{proof}{\begin{trivlist}\item[]{\em Proof:}\/}{\qed\end{trivlist}}
\newenvironment{proofof}[1][Proof]{\noindent \textit{#1.} }{\ \qed}

\newcommand{\R}{{\mathbb R}}
\newcommand{\C}{{\mathbb C\hspace{0.05 ex}}}
\newcommand{\N}{{\mathbb N}}

\newcommand{\ep}{\varepsilon}

\DeclareMathOperator{\Tr}{Tr}

\newcommand{\vphi}{{\varphi}}

\newcommand{\be}{\begin{equation}}
\newcommand{\ee}{\end{equation}}
\newcommand{\bes}{\begin{equation*}}
\newcommand{\ees}{\end{equation*}}
\newcommand{\beqs}{\begin{eqnarray}}
\newcommand{\eeqs}{\end{eqnarray}}

\newcounter{jlisti}

\title{Self-similar asymptotics for a modified Maxwell-Boltzmann equation in systems subject to deformations}
\author{Alexander Bobylev, 
 Alessia Nota, 
  Juan J.~L.~Vel\'azquez }

\date{\today}

\def\adresse{
\begin{description}
\item[Alexander V. Bobylev:] Keldysh Institute of Applied Mathematics, RAS \\ Miusskaya Pl.,4, Moscow, Russian Federation \\
E-mail: \texttt{alexander.bobylev@kau.se}

\item[Alessia Nota:] Institute for Applied Mathematics,\\ University of Bonn, Endenicher Allee 60, 53115 Bonn, Germany\\
E-mail: \texttt{nota@iam.uni-bonn.de}

\item[Juan J. L. Vel\'azquez:] Institute for Applied Mathematics,\\ University of Bonn, Endenicher Allee 60, 53115 Bonn, Germany\\ 
E-mail: \texttt{velazquez@iam.uni-bonn.de}
\end{description}
}

\begin{document}

\maketitle

\abstract{
In this paper we  study a generalized class of Maxwell-Boltzmann equations which in addition to the usual collision term contains a linear deformation term described by a matrix A. This class of equations arises, for instance, from the analysis of  homoenergetic solutions for the Boltzmann equation considered by many authors since 1950s. Our main goal is to study a large time asymptotics of  solutions under assumption of smallness of the matrix A.  The main  result of  the paper is formulated in Theorem 2.1. Informally stated, this Theorem says that,  for sufficiently small norm of A,  any non-negative solution with finite second moment tends to a self-similar solution of relatively simple form for large values of time. This is what we call ``the self-similar asymptotics".  We also prove that the higher order moments of the  self-similar profile are finite under further smallness condition on the matrix A.
}

\bigskip

\bigskip

\tableofcontents

\section{Preliminaries} \label{sec:1}

Let $f(v,t)$ be a one-particle distribution function, $v\in \R^d$ and $t\in \R_{+}=[0,\infty)$ denote respectively the particle velocity and time, $d\geq 2$. We consider the following modified Boltzmann equation for pseudo-Maxwell molecules 
\be \label{eq:GenBolt}
\partial_{t} f- \text{div}_{v} \left(A v f \right)    =Q\left(f,f \right) 
\ee
where $A\in M_{d\times d}(\R)$ is a real matrix and $Q\left(f,f\right)$ is the collision integral
\be  \label{eq:CollBolt}
Q\left(f,f\right) \left(  v\right)= \int_{\mathbb{R}^{d}}dw \int_{S^{d-1}}dn \ g\left(  \frac{(v-w)}{\vert v-w\vert } \cdot n \right)\left[  f(v^{\prime})f(w^{\prime})- f(v) f(w)\right]
\ee 
with $n\in S^{d-1}$, being $S^{d-1}$ the unit sphere in $\mathbb{R}^{d}$. Here $(v,w)$ is a pair of velocities in incoming collision configuration. We will denote by convenience $u=v-w$ the relative velocity, $\hat{u}=\frac{u}{\vert u \vert }$, and the corresponding pair of outgoing velocities is defined by the collision rule
\begin{align*} 
v^{\prime}  &  =\frac{1}{2} \left(v+w+\vert u \vert n \right), \\
w^{\prime}  &  =\frac{1}{2}\left(v+w-\vert u \vert n \right).  
\end{align*}

We will consider the initial value problem such that
\be \label{eq:incdt}
f(v,0)=f_0(v)\geq 0, \qquad \int_{\R^d} dv \ f_0(v)=1.
\ee 
The normalization in \eqref{eq:incdt} can be assumed without any loss of generality. Similarly, we can assume that the kernel $g(\eta)$ in \eqref{eq:CollBolt}, with $\eta\in[-1,1]$, is non-negative and normalized as 
\be \label{eq:kernel}
\int_{S^{d-1}}dn \ g(\omega\cdot n)=1, \quad \text{for any}\quad \omega\in  S^{d-1}.
\ee 
Thanks to the normalization assumptions both $f(v,t)$ and $g( \hat{u} \cdot n)$ are probability densities. 
\medskip

Our aim is to study the long-time behavior of $f(v,t)$. To be more precise, we intend 
\begin{itemize}
\item[(a)] to construct a class of self-similar solutions to \eqref{eq:GenBolt}-\eqref{eq:CollBolt} having the form
\be \label{eq:ssprof}
f(v,t)=e^{-d\beta t} F(v e^{-\beta t}),\; \beta\in\R,
\ee
\item[(b)] to prove that these solutions describe a large time asymptotics for a certain class of initial data. 
\end{itemize}
We observe that $\beta\in\R$ in \eqref{eq:ssprof} has to be thought as an eigenvalue. More precisely, solutions to \eqref{eq:GenBolt} with the form \eqref{eq:ssprof} exist only for specific choices of $\beta$ depending on the matrix $A$.  Therefore, this problem can be understood as an example of self-similarity of the second kind (see \cite{Bar}).

We are also interested in the self-similar profile $F(v)$, in particular in power moments of this function.

\medskip
Our motivation to analyze the properties of the solutions to the equation (\ref{eq:GenBolt}) arises from the study of the so-called
\textit{homoenergetic solutions.} These are some particular solutions of the
Boltzmann equation which were introduced in \cite{G1} and
\cite{T}. These solutions $g=g\left(  x,v,t\right)  $ solve the classical spatially inhomogeneous Boltzmann equation and they have the particular form $g\left(
x,v,t\right)  =f\left(  v-A\left(  t\right)  x,t\right)  $ for some suitable
choices of the matrix $A\left(  t\right)  \in M_{d\times d}\left(
\mathbb{R}\right)  .$ A detailed description of the possible choices of
$A\left(  t\right)  $ can be found in
\cite{JNV}. In order to obtain a solution $g$
of the Boltzmann equation, the matrix $A\left(  t\right)  $ must solve an
equation that can be reduced to (\ref{eq:GenBolt}) by means of some suitable
changes of variables.

The properties of the homoenergetic solutions have been intensively studied in
the physical literature, mostly in the case of Maxwell molecules (cf.
\cite{GS}), while their mathematical properties have been less studied. The global well-posedness of
(\ref{eq:GenBolt}) has been proved in \cite{CercArchive}, 
\cite{Cerc2000}, \cite{Cerc2002}. Additional properties concerning homoenergetic solutions in the two dimensional case can be found in \cite{BobCarSp}. More recently, a systematic analysis of 
the mathematical properties and the long-time asymptotics of the homoenergetic
solutions has been provided in several papers both for Maxwell molecules as
well as for more general class of collision kernels (cf. \cite{JNV}, \cite{JNV2}, \cite{JNV3}, \cite{MT}). In particular, the existence of self-similar solutions to
(\ref{eq:GenBolt}) with the form (\ref{eq:ssprof}) has been proved in
 \cite{JNV} if $A$ is small in some suitable norm.

In this paper we will obtain a different proof from the one in 
\cite{JNV} of the existence of
self-similar solutions of (\ref{eq:GenBolt}) with the form (\ref{eq:ssprof})
for small $A$ and in addition we will prove that these solutions are unique (in the class of probability
measures). We will also prove that they are stable and, more precisely, that
they attract all the solutions of (\ref{eq:GenBolt}) having the same mass
as well as the same first order moments. Actually, we will prove that they are exponentially convergent in the topology of the uniform convergence for the Fourier transforms of the probability measures. The main tool that we will
use to prove these results is the well developed machinery available for the
study of the Boltzmann equations in the case of Maxwell molecules by means of the
Fourier transform method that was introduced in \cite{Bo75}.

\bigskip

We finally notice that self-similar solutions to the homogeneous Boltzmann
equation (i.e.~the equation (\ref{eq:GenBolt}) with $A=0$) have been
considered in several situations that we discuss here. In the case of the
classical Boltzmann collision kernel and Maxwell molecules, the existence and stability of
self-similar solutions with the form \eqref{eq:ssprof} was proved in \cite{BC02a}, \cite{BC02b}, \cite{BC03}. These solutions have infinite second
order moments, something that might be expected, since otherwise those
solutions would violate the conservation of energy property. On the other
hand, the existence  and stability of self-similar solutions with finite second order moments
to the homogeneous Boltzmann equation for granular media and Maxwell molecules
was conjectured in \cite{EB02} and rigorously proved in
\cite{BCT03}. Such solutions have also the form  \eqref{eq:ssprof}. They will not be
discussed in this paper, but we remark that they can exist because
conservation of energy does not hold in this case. 
We further remark that all the previous self-similar solutions obtained for the case $A=0$ (both for classical Boltzmann and for granular media) are radially symmetric while the solutions obtained in this paper are non symmetric in general because the matrix $A$ breaks the invariance under rotation of the problem.

\section{Strategy of the paper and main results }  
\label{sec:2}

It is well known that the collision operator \eqref{eq:CollBolt} can be simplified by the Fourier Transform (see \cite{Bo75, Bo88}).  On the other hand the Fourier Transform of the hyperbolic term $div_v(Af)$ has a simpler form. 
More precisely, we denote by $\varphi(k)=\hat{f}(k)$ the Fourier Transform of the one particle probability density $f(v)$, i.e.
\be \label{eq:phif}
\vphi(k)=\hat{f}(k)=\int_{\R^d} dv \ f(v)e^{-i k\cdot v}, \;\; k\in\R^d
\ee
and, using \eqref{eq:kernel}, we obtain the following equation
\be \label{eq:BoltFourier}
\partial_t \vphi  + \left(Ak\right)\cdot \partial_k \vphi = \mathcal{I}^{+}(\vphi,\vphi)-\vphi_{|_{k=0}} \vphi 
\ee 
where 
\be \label{eq:BoltCollFourier}
 \mathcal{I}^{+}(\vphi,\vphi)(k)= \int_{S^{d-1}}dn \ g\left(  \hat{k} \cdot n \right)\vphi(k_{+}) \vphi (k_{-}) \;\; \text{with}\;\; k_{\pm}=\frac{1}{2}\left( k\pm \vert k \vert n\right),\quad \hat{k}=\frac{k}{\vert k\vert}.
\ee
\smallskip

\noindent The derivation of this formula can be found e.g.~in \cite{Bo75}, \cite{Bo88}. 
The initial condition \eqref{eq:incdt} becomes
\be \label{eq:incdtFourier}
\vphi (k,0)=\vphi_0(k)= \int_{\R^d} dv \ f_0(v)e^{-i k\cdot v}, \quad \vphi_0(0)=1.
\ee
Note that  \eqref{eq:BoltFourier} implies the mass conservation condition
\be \label{eq:phifk0}
\vphi(0,t)=\vphi_0(0)=1.
\ee 
Hence, equation \eqref{eq:BoltFourier} reads 
\be \label{eq:BoltFourier2}
\partial_t \vphi  +\vphi + \left(Ak\right)\cdot \partial_k \vphi = \mathcal{I}^{+}(\vphi,\vphi) . 
\ee 
For a given  matrix $A\in M_{d\times d}(\R)$ we will use the following classical matrix norm: 
\be \label{eq:normA2}
\vert |A\vert |=\sup_{\vert k\vert =1} \vert A k\vert , 
\ee
where $|\cdot |$ denotes the Euclidean norm in $\R^d$.

We recall that various generalizations of the equation \eqref{eq:BoltFourier2} with $A=0$ and with radial solutions $\vphi ( k  ,t)=\vphi ( \vert k \vert ,t)$ were studied in detail in \cite{BCG09} where exactly the same questions on self-similar asymptotics were considered. We shall see below how the approach of that paper can be generalized to \eqref{eq:BoltFourier2} with $A\neq 0$ and without the assumption that $\vphi=\vphi ( \vert k \vert ,t)$.   
First we need to introduce a suitable space of functions for the Cauchy problem  \eqref{eq:incdtFourier}, \eqref{eq:BoltFourier2}. A natural space to solve the modified Boltzmann equation \eqref{eq:GenBolt} is the space of probability measures $f(\cdot,t)$ in $\R^d$ depending on the parameter $t\in \R_{+}$. Then, it is natural to solve  \eqref{eq:incdtFourier}, \eqref{eq:BoltFourier2}  in the space of Fourier transforms $\vphi (\cdot,t)$ of time-dependent probability measures $f(\cdot,t)$, with $t\in \R_{+}$. We recall that in Probability Theory the Fourier transforms of probability measures are called characteristic functions.  The properties of these functions are well known (see, e.g., \cite{Fe2}). They are a subset $\Phi$ of the space of complex-valued continuous functions $C(\R^d;\mathbb{C})$. Then
\be\label{def:setffi}
\Phi=\{\varphi\in C(\R^d;\mathbb{C})\; :\; \varphi(k)= \hat{f}(k), \quad \text{for}\quad f\in \mathcal{P}_{+}(\R^d)  \},
\ee 
where we denote as $\mathcal{P}_{+}(\R^d)$ the set of Radon probability measures in $\R^d$, $d\geq1$, and the Fourier transform $\hat{f}(\cdot)$ is defined as in \eqref{eq:phif}.  Notice that  \eqref{eq:phif} implies that for every $\vphi \in \Phi$ we have $\|\vphi\|_{\infty}\leq \vphi(0)=1$ where $\| \vphi\|_{\infty}=\sup_{k\in\R^d} |\vphi(k)|$. Therefore, the set $\Phi$ is contained in the subset $\mathcal{U}$ of the unit sphere of $C(\R^d;\mathbb{C})$ defined by means of 
\be \label{def:setU}
\mathcal{U}:= \{\psi \in C(\R^d;\mathbb{C}) \, : \; \vert| \psi\vert|_{\infty}=1, \quad \psi(0) =1 \}.
\ee
It is well-known that the set $\Phi$ (cf. \eqref{def:setffi}) can be characterized alternatively by means of Bochner's Theorem (see for instance \cite{Fe2} Chapter XIX) but the definition of $\Phi$ above will enough for our purposes.

It follows from \eqref{def:setffi} that $\Phi$ is a convex set, i.e. if $\vphi_{1,2}\in \Phi$ then $\alpha\vphi_1+(1-\alpha)\vphi_2 \in \Phi$ for any $\alpha\in [0,1]$.

 We will use repeatedly the following family of functional spaces. Given $T\in (0,\infty]$ we define the space $C([0,T); \Phi)$ as the set of  functions continuous in time $\psi(\cdot,t)\in \Phi$ for every $t\in [0,T)$ and with $\Phi$ endowed with the topology given by $\|\cdot\|_{\infty}.$ We introduce a topology in $C([0,T); \Phi)$ by means of the metric 
 \be \label{def:distT}
 d_{T}(\psi_1,\psi_2)=\sup_{0\leq t\leq T} \vert| \psi_1(\cdot,t) -\psi_2(\cdot,t) \vert|_{\infty}.
 \ee 
 Notice that the space  $C([0,T); \Phi)$ is a complete metric space due to the fact that the space $\Phi$ is closed in the topology of the uniform convergence (cf. \cite{Fe2}).

\medskip

The goal of this paper is to obtain the following results:
\begin{enumerate}
\item To show the existence and uniqueness of solutions to the time-dependent problem   \eqref{eq:incdtFourier}, \eqref{eq:BoltFourier2} in the class of characteristic functions $\Phi$. 
\item To prove that for $A$ sufficiently small there exists a self-similar solution to \eqref{eq:BoltFourier2} with the form 
\be \label{eq:ssprofF}
\varphi(k,t)= \Psi(ke^{\beta t}),
\ee
for some suitable $\beta\in \R$.  This form of the characteristic function is associated to a probability measure as in \eqref{eq:ssprof}.

\item To prove that for every $A$ for which a self-similar solution as in \eqref{eq:ssprofF} exists, 
the solution $\varphi(k,t)$ to the time-dependent problem  \eqref{eq:incdtFourier}, \eqref{eq:BoltFourier2} satisfies 
\be \label{eq:cvffipsi}
\lim_{t\to\infty} \varphi(e^{-\beta t}k,t)= \Psi(k)
\ee
assuming suitable conditions for the moments associated to the initial value $\vphi(k,0)=\vphi_0(k)$. 
The convergence in \eqref{eq:cvffipsi} is understood in the sense of uniform convergence, i.e. in the norm $\|\cdot\|_{\infty}$. 

Notice that \eqref{eq:cvffipsi} implies uniqueness of the self-similar solution obtained in item $2.$
\item To study properties of the distribution function $f(v)$ associated to the characteristic function $\Psi$ defined in \eqref{eq:ssprofF}. 
\end{enumerate}

Some of the questions above have been addressed in  \cite{JNV}. More precisely, the well-posedness result for the time-dependent problem in item $1.$ follows from Theorem 4.2 in \cite{JNV} but an alternative proof, using characteristic functions, is included here for reader's convenience. The existence of self-similar solutions with the form \eqref{eq:ssprofF} (or more precisely as in \eqref{eq:ssprof}) has been obtained in \cite{JNV} under the same assumptions. The main novelty in this paper, that is the content of item $3.$, is to prove uniqueness and stability of the self-similar solution described in item $2.$ More precisely, the main result of the paper is the following.

\medskip

\begin{theorem} \label{thm}
Suppose that $f\in C([0,\infty);\mathcal{P}_{+}(\R^d))$ is a solution to \eqref{eq:GenBolt}, \eqref{eq:CollBolt} with initial datum  $f(\cdot,0)=f_0(\cdot)\in \mathcal{P}_{+}(\R^d)$ such that $\int_{\R^d} dv\ f_0(v) |v|^p<\infty$ for $p>2$.  There exists an $\ep_0=\ep_0(d,p)>0$ such that if  $\|A\|\leq \ep_0$ there exists a measure $f_{st}\in \mathcal{P}_{+}(\R^d)$ satisfying $\int_{\R^d} dv\ f_{st}(v) |v|^p<\infty$ with $2<p\leq 4$, and $\beta=\beta(A)\in\R$ such that 
\be
\left( e^{\beta t} \right)^d \ f \left(e^{\beta t}v+ e^{-tA^T} U,t\right) \rightarrow  \lambda^{-d}f_{st}\big({\lambda}^{-1} v\big) \quad \text{as}\quad t\to\infty
\ee
in the weak topology of $\mathcal{P}_{+}(\R^d)$ and where $\lambda=\lambda(f_0,A)> 0$, $U\in\R^d$. Here we denote as $A^T$ the transpose of $A$.  Moreover, for every $M\in\N$ there exists an $\ep_M\in (0,\ep_0)$ sufficiently small such that if $\| A\|\leq \ep_M $  then $\int_{\R^d} dv \ f_{st}(v)|v|^M<\infty$.
\end{theorem} 

More details concerning the  specific form of the measure $f_{st}$, the vector $U$ and $\lambda$ will be given in Theorem \ref{th:exssf} and Theorem \ref{thm:Main_f}.

Notice that Theorem \ref{thm} implies a global stability result for the self-similar solutions $f_{st}$. Furthermore, we remark that since the  detailed-balance does not hold in general  for \eqref{eq:GenBolt}, the proof cannot rely only on entropy arguments, as those used in \cite{ToVi} in the case $A=0$.  The argument used in this paper to prove Theorem \ref{thm} is a perturbative argument which strongly relies on the particular form of the collision operator for Maxwellian molecules. \medskip

The plan of the paper is the following. In Section \ref{sec:3} we recall the notion of $\mathcal{L}$-Lipschitzianity for the operator $\mathcal{I}^{+}$, which has been introduced in \cite{BCG09}.  In Section \ref{sec:4} we prove a well-posedness result for the Cauchy problem \eqref{eq:incdtFourier}, \eqref{eq:BoltFourier2}. Moreover, we also prove a fundamental Comparison Theorem which allows to estimate the difference of two solutions to  \eqref{eq:incdtFourier}, \eqref{eq:BoltFourier2}. This Theorem is used in Section \ref{sec:5} to show that two solutions  to  \eqref{eq:incdtFourier}, \eqref{eq:BoltFourier2}, with a sufficiently high order of tangency at $k=0$, approach to each other exponentially fast. In Section \ref{sec:6} we collect several properties of the self-similar solutions which are expected to yield the long-time asymptotics of solutions to   \eqref{eq:incdtFourier}, \eqref{eq:BoltFourier2}. In particular, we show that using a suitable change of variables, the problem can be reduced to the analysis of the steady states for  \eqref{eq:BoltFourier2} with particular choices of the matrix $A$.  The existence and uniqueness of these steady states is proved in Section \ref{sec:7} and their stability is obtained in Section \ref{sec:8}. Moreover, in Section  \ref{sec:9}, we prove that the steady state has moments of any given order $M>2$ under a smallness assumption on $\|A\|$, namely $\|A\| \leq \ep_{M}$.
 In Section \ref{sec:10} we reformulate the results obtained for characteristic functions $\vphi$ in terms of the probability measures $f$. 
 In Section \ref{sec:11} we summarize the main results obtained in this paper and we discuss some further perspectives.  In Appendix A we derive the matrix equation for the second order moments which is used in Sections \ref{sec:6}-\ref{sec:8}.

\section{$\mathcal{L}-$Lipschitzianity}\label{sec:3}

We will  assume in the rest of the paper the following assumptions for the collision kernel $g(\eta),$ $\eta\in[-1,1]$, 
 associated to the Boltzmann operator  \eqref{eq:CollBolt}: 
 \be\label{ass:kernel}
 g:[-1,1]\rightarrow \R_{+}, \;\; g\; \text{measurable} \quad \text{and} \quad \int_{S^{d-1}} dn \ g(e\cdot n)<\infty\;\; \text{with}\; e=(1,0,0,\dots) .
 \ee
 Notice that the invariance of the scalar product under the orthogonal group implies that $\int_{S^{d-1}} g(\omega \cdot n)dn=\int_{S^{d-1}} g(e\cdot n)dn$ for any $\omega\in S^{d-1}$. We will assume also the normalization condition \eqref{eq:kernel}. 
 
We now prove a crucial regularity property of the quadratic operator $\mathcal{I}^{+}(\vphi,\vphi)$ defined in \eqref{eq:BoltCollFourier} acting on the space $ C(\R^d;\mathbb{C})$. 

We first introduce the linearized operator 
\be \label{eq:linBolzF}
\mathcal{L}\big[\vphi\big](k)=\mathcal{I}^{+}(\vphi,1)+\mathcal{I}^{+}(1,\vphi)=  \int_{S^{d-1}}dn \ g\left(  \hat{k} \cdot n \right)\left[\vphi(k_{+})+ \vphi (k_{-})\right]
\ee
where $k_{\pm}$ and $\hat{k}$ are as in \eqref{eq:BoltCollFourier}.  Notice that the operator $\mathcal{I}^{+}$ considered as a bilinear form is non symmetric. 

We then collect the main properties of the operator $\mathcal{I}^{+}$ in the following lemma which is a straightforward adaptation of the one proved in \cite{BCG09}.

\begin{lemma}\label{lm:bdI}
Suppose that $g$ satisfies \eqref{eq:kernel} and \eqref{ass:kernel}. Then \eqref{eq:BoltCollFourier} defines an operator $\mathcal{I}^{+}: \mathcal{U}\times\mathcal{U} \rightarrow\mathcal{U} $ where $\mathcal{U}$ is the subset of $C(\R^d;\C)$ defined in \eqref{def:setU}. Analogously, $\mathcal{L}$ in \eqref{eq:linBolzF} defines an operator $\mathcal{L}: C(\R^d;\C)\rightarrow C(\R^d;\C)$. Moreover, for any $\vphi_j(\cdot) \in C(\R^d;\mathbb{C})$, $j=1,2$ such that $\vert| \vphi_j\vert|_{\infty}\leq 1$ the following estimate holds
\be \label{eq:est}
\vert \mathcal{I}^{+}(\vphi_1,\vphi_1)-\mathcal{I}^{+}(\vphi_2,\vphi_2)\vert (k) \leq \mathcal{L} \big[\vert \vphi_1-\vphi_2\vert \big] (k),
\ee
for any $k\in \R^d$.
\end{lemma}

\begin{proof} We first observe that \eqref{eq:BoltCollFourier} and \eqref{eq:linBolzF} define a mapping from the space of continuous functions to itself. The only difficulty is to check the continuity at $k=0$ since $\hat{k}$ is not well-defined for $k=0$. However, the operators can be given a meaning in the sense of limit as $k\to 0$ using the fact that then $\vphi(k_{\pm})$ can be approximated by $\vphi(0)$ and the integral 
$\int_{S^{d-1}}dn \ g (  \hat{k} \cdot n )$ is independent on $\hat{k}$. 

In order to prove that $\mathcal{I}^{+}$ defines a quadratic operator from $ \mathcal{U}\times\mathcal{U} $ into $ \mathcal{U}$ we first notice that
\be \label{eq:estI}
\vert| \mathcal{I}^{+}(\vphi_1,\vphi_2) \vert|_{\infty} \leq \vert| \vphi_1 \vert|_{\infty} \vert| \vphi_2\vert|_{\infty} =1 \quad \text{for any } \ \vphi_1,\vphi_2\in \mathcal{U}. 
\ee
In order to prove that the operator  $\mathcal{I}^{+}(\vphi,\vphi)$ maps the set $\mathcal{U}$ 
into itself it only remains to prove that $\mathcal{I}^{+}(\vphi,\vphi)(0)=1$. This immediately follows from \eqref{eq:BoltCollFourier} using the fact that $\vphi(0)=1$. 

In order to prove \eqref{eq:est} we argue as follows. Given $\vphi_1,\vphi_2\in \mathcal{U}$ we have  
$$\mathcal{I}^{+}(\vphi_1,\vphi_1)-\mathcal{I}^{+}(\vphi_2,\vphi_2)=\frac 1 2 \left[ \mathcal{I}^{+}(\vphi_1+\vphi_2,\vphi_1+\vphi_2)-\mathcal{I}^{+}(\vphi_1-\vphi_2,\vphi_1-\vphi_2)\right].$$
Using that $\vert\vphi_1+\vphi_2 \vert \leq 2$ as well as $\pm (\vphi_1-\vphi_2)\leq \vert \vphi_1-\vphi_2 \vert $  we obtain, using also \eqref{eq:linBolzF}, 
$$\vert \mathcal{I}^{+}(\vphi_1,\vphi_1)-\mathcal{I}^{+}(\vphi_2,\vphi_2)\vert (k)\leq 
\mathcal{L} \big[\vert \vphi_1-\vphi_2\vert \big] (k)$$
for any $k\in\R^d$, whence \eqref{eq:est} follows.
\end{proof}

The class of operators satisfying \eqref{eq:est} were termed as ``$\mathcal{L}-$Lipschitz operators" in \cite{BCG09} due to the fact that they satisfy a Lipschitz condition with respect to the linear operator $\mathcal{L}$ instead of with respect to a norm. Notice that the condition \eqref{eq:est} is a pointwise condition for any $k\in\R^d$.  
This condition is much stronger than the classical Lipschitz condition 
\be \label{eq:ILip}
\vert| \mathcal{I}^{+}(\vphi_1,\vphi_1)-\mathcal{I}^{+}(\vphi_2,\vphi_2)\vert|_{\infty} \leq 2 \vert| \vphi_1-\vphi_2\vert|_{\infty},
\ee
which follows in a straightforward way from \eqref{eq:est}. 
We shall see below that the property \eqref{eq:est} will play a crucial role in the proof of the results presented in Section \ref{sec:2}.

We now state an additional property of the operator $\mathcal{I}^{+}$ that will be useful in order to prove that the solutions $\vphi$ of  \eqref{eq:incdtFourier}, \eqref{eq:BoltFourier2}  are characteristic functions for every $t\in R_{+}$. 
\begin{lemma}\label{eq:I+map}
Let $\Phi$ and  $\mathcal{I}^{+}$ be as in \eqref{def:setffi} and \eqref{eq:BoltCollFourier} respectively. Then, for any $\vphi\in \Phi$ we have 
$\mathcal{I}^{+}(\vphi,\vphi)\in\Phi$.
\end{lemma}
\begin{proof}
Since $\vphi\in \Phi$ we have that $\vphi=\hat{f}$ for some $f\in \mathcal{P}_{+}(\R^d)$. Then, the inverse Fourier transform of $\mathcal{I}^{+}$ is 
\be\label{eq:idFourier}
\frac {1}{(2\pi)^d} \int_{\R^d}dk \ \mathcal{I}^{+}(\vphi,\vphi)(k)\, e^{i k\cdot v} = \int_{\R^d}dw \int_{S^{d-1}} dn\ g( \hat{u}\cdot n)f(w') f(v')=\mu(v).
\ee
This identity has to be understood in the sense of distributions, acting on suitable test functions. More precisely, we have 
\bes
\frac {1}{(2\pi)^d} \int_{\R^d}dk \ \hat{\psi}(k) \mathcal{I}^{+}(\vphi,\vphi)(k)\, e^{i k\cdot v} = \int_{\R^d}dv \ \Psi(v)\mu(v),\quad \forall \; \psi\in \mathcal{S}(\R^d)
\ees
where $\mathcal{S}(\R^d)$ denotes the space of Schwartz functions on $\R^d$. 
Notice that the distribution $\mu$ defined in \eqref{eq:idFourier} is a a non-negative Radon measure.  Moreover, it is a probability measure since
\begin{align*}
\int_{\R^d}dv \int_{\R^d}dw \int_{S^{d-1}} dn\ g( \hat{u}\cdot n)f(w') f(v')&= \int_{\R^d}dv' f(v') \int_{\R^d}dw' f(w') \int_{S^{d-1}} dn\ g( \hat{u}\cdot n)\\&
 = \left(\int_{\R^d}dv f(v)\right)^2,
\end{align*}
where in the last identity we used \eqref{eq:kernel}. 
\end{proof}

\section{Existence and uniqueness of solutions. Comparison Theorem}  
\label{sec:4}

We reformulate the Cauchy problem  \eqref{eq:incdtFourier}, \eqref{eq:BoltFourier2} as an integral equation using Duhamel's formula, or equivalently, we rewrite  \eqref{eq:incdtFourier}, \eqref{eq:BoltFourier2} in mild form.  We will also use by shortness the following notation
\be \label{eq:notshort}
\Gamma\big[\vphi\big]= \mathcal{I}^{+}(\vphi ,\vphi),\qquad \hat{D}=\big(Ak\big)\cdot \partial_k .
\ee
Then  \eqref{eq:incdtFourier}, \eqref{eq:BoltFourier2} yields, formally, the following integral equation 
\be \label{eq:IntBoltFou}
\vphi (\cdot, t)= E(t)\vphi_0(\cdot) + \int_0^t d\tau \ E(t-\tau) \Gamma\big[\vphi (\cdot,\tau)\big]
\ee
where 
\be \label{eq:defE}
E(t)=\exp\big(-t (1+\hat{D}) \big).
\ee 
Notice that the linear operator $E(t)$ acts on any function $\psi\in C(\R^d;\C)$ in the following way:
\be \label{eq:opE}
E(t)\psi(k) = e^{-t} \psi \left(e^{-At}k\right), \quad k\in\R^d, \; t\in\R.
\ee
The concept of solutions to  \eqref{eq:incdtFourier}, \eqref{eq:BoltFourier2}  that we will use in this paper are functions  $\vphi \in C\left( [0,+\infty); \Phi\right)$ solving the integral equation \eqref{eq:IntBoltFou}.

We will need some auxiliary results. 
\begin{lemma}\label{lm:bdS}
Let $q=q(k)$ be a  continuous mapping $q: \R^d\to \R^d$. We define a linear operator  $S: C(\R^d;\C)\rightarrow C(\R^d;\C)$ by means of 
\be \label{eq:defSpsi}
S\psi(k)=\psi(q(k)), \quad \forall \; k\in\R^d, \quad \forall \;\psi\in C(\R^d;\C).
\ee
Then, for any $\vphi_{1,2}(k)\in \mathcal{U}$ such that $\|\vphi_{1,2}\|\leq 1$ we have
\be \label{eq:estS}
\vert S\Gamma\big[\vphi_1\big]-S\Gamma\big[\vphi_2\big] \vert  (k) \leq S \mathcal{L}\left[ \vert \vphi_{1}-\vphi_{2} \vert \right] (k),\quad \text{for any} \quad k\in\R^d .
\ee
\end{lemma}
\begin{proof}
Using the definition of the operator $S$ in \eqref{eq:defSpsi}, as well as \eqref{eq:est} in  Lemma \ref{lm:bdI}, we obtain
\begin{align*}
\vert S\Gamma\big[\vphi_1\big]-S\Gamma\big[\vphi_2\big] \vert  (k)&=\vert S \left(\Gamma\big[\vphi_1\big]-\Gamma\big[\vphi_2\big] \right)\vert  (k) = \vert \Gamma\big[\vphi_1\big]-\Gamma\big[\vphi_2\big]\vert  (q(k))\\& \leq \mathcal{L}\left[\vert \vphi_{1}-\vphi_{2} \vert\right](q(k)) = S \mathcal{L}\left[\vert \vphi_{1}-\vphi_{2}\vert\right] ( k)
\end{align*}
whence the result follows. 
\end{proof}
\begin{remark}
Notice that the function $q: \R^d\to \R^d$ in Lemma \eqref{lm:bdS} is continuous but not necessarily bounded. 
\end{remark}

We introduce some lemmas to prove directly that the equation is solved in the class of characteristic functions.
\begin{lemma}\label{lem:qBk}
Suppose that $q:\R^d\rightarrow \R^d$ is a linear mapping $q(k)=Bk$ with $B\in M_{d\times d}(\R)$ such that $\vert det(B)\vert\neq 0$ and $k\in\R^d$. Let the operator $S$ be as in \eqref{eq:defSpsi}. Then $S(\Phi)\subset\Phi$ with the set $\Phi$ as in \eqref{def:setffi}.
\end{lemma}
\begin{proof} 
We recall that $\vphi=\hat{f}$. We have
\begin{align*}
\left(S\vphi\right)(k)&=\vphi(q(k))=\vphi(Bk)=\int_{\R^d}dv \  f(v)e^{-i (Bk)\cdot v} \\&=\int_{\R^d} dv \ f(v)e^{-i k\cdot (B^T)v} =\int_{\R^d}  dw\ f((B^T)^{-1}w)e^{-i k\cdot w} det((B^T)^{-1})
\end{align*}
where in the last identity we used the change of variable $w=B^Tv$ as well as  $det((B^T)^{-1})=\frac{1}{det(B^T)}=\frac{1}{det(B)}$.
 
We need to prove that the Radon measure $g$ defined by means of $g=\frac{1}{det(B^T)} f\circ \tilde{q}$ is a probability measure, namely  $g\in \mathcal{P}_{+}(\R^d)$. 
This follows from the fact that 
\bes
\int_{\R^d} dw \ g(w)=\int_{\R^d} dw \ \frac{1}{det(B^T)} f((B^T)^{-1}w)=\frac{1}{det(B^T)}\int_{\R^d} dv \ f(v) det(B^T).
\ees
Hence $S(\Phi)\in \Phi$.
\end{proof}

\begin{lemma}\label{lem:Int}
Let be $\psi_{1},\ \psi_{2} \in C([0,\infty); \Phi)$. Then, for any $t\in [0,\infty)$ we have
\be 
h(\cdot, t)=e^{-t} \psi_1(\cdot,t)+\int_0^t d\tau e^{-(t-\tau)} \psi_2(\cdot,\tau) \in \Phi.
\ee
\end{lemma}
\begin{proof} 
Since $\psi_1(\cdot,t)=\hat{f}_1(\cdot, t)$, $\psi_2(\cdot,t)=\hat{f}_2(\cdot, t)$ with $f_1, f_2 (\cdot, t) \in \mathcal{P}_{+}(\R^d)$ for any $t\in [0,T)$ then 
we have that $h(\cdot, t)=\hat{g}(\cdot,t)$ for $t\in [0,T)$ with
$${g}(\cdot, t)=e^{-t} f_1(\cdot,t)+\int_0^t d\tau e^{-(t-\tau)} f_2(\cdot,\tau).$$ 
Thus $g(\cdot, t)\in \mathcal{M}_{+}(\R^d)$. To prove that $g(\cdot, t)\in \mathcal{P}_{+}(\R^d)$ we just notice that, since $f_1,\ f_2 (\cdot, t) \in \mathcal{P}_{+}(\R^d)$, 
$$\int_{\R^d} g(v,t)dv=e^{-t}+\int_0^t d\tau e^{-(t-\tau)}=1$$
whence the result follows.
\end{proof} 

The main result that we prove in this section is the following.
\begin{theorem}\label{thm:1}
Let be $\vphi_0(k)\in \Phi$ where $\Phi$ is defined as in \eqref{def:setffi}. Then there exists a unique $\vphi \in C([0,\infty); \Phi)$ safisfying \eqref{eq:IntBoltFou}. 
\end{theorem}

\begin{proof} 
We will prove the result using a standard Banach fixed point argument. 

We define an operator $\mathcal{T}: C([0,T]; \Phi)\rightarrow C([0,T]; \Phi)$ by means of
\be
\left(\mathcal{T}\vphi \right)(\cdot, t)=E(t)\vphi_0+\int_0^t d\tau e^{-(t-\tau)}\Gamma[\vphi(\tau)].
\ee
In order to check that the operator $\mathcal{T}$ maps the space $ C([0,T]; \Phi)$ into itself, we first notice that $\Gamma[\Phi]\subset\Phi$ thanks to Lemma \ref{eq:I+map}. We rewrite the operator $E(t)$ as $E(t)=e^{-t}e^{t\hat{D}}$ then we have
\be
\left(e^{t\hat{D}} \vphi \right) (k) =\vphi(q(k))\quad \text{with}\quad  q(k)=e^{-tA} k.
\ee 
Hence, using Lemma \ref{lem:qBk}, it follows that $\left(e^{t\hat{D}} \vphi_0 \right)\in\Phi$ for any $t\in [0,T)$ as well as $e^{-(t-\tau)}\Gamma[\vphi(\tau)]\in \Phi$ for any $0\leq \tau\leq t < T$. 
It is straightforward to see that $\left(e^{t\hat{D}} \vphi_0 \right)$ and $e^{-(t-\tau)}\Gamma[\vphi(\tau)]$ are continuous in time. 
Therefore, using Lemma \ref{lem:Int}, we obtain that $\mathcal{T}\vphi \in C([0,T);\Phi)$.  
Notice that, since $\vphi, \mathcal{T}\vphi \in C([0,T);\Phi)$ then $\vert | \vphi (\cdot,t) \vert|_{\infty}=\vert | \mathcal{T}\vphi (\cdot,t) \vert |_{\infty}=1$ for any $t\in [0,T)$.

In order to prove that the operator $\mathcal T$ is contractive we consider $\vphi_1, \vphi_2 \in C([0,T);\Phi)$ and look at 
$$
\mathcal T\vphi_1 (\cdot,\tau)- \mathcal T\vphi_2 (\cdot,\tau)=\int_0^t d\tau \ E(t-\tau) \big\{\Gamma\big[\vphi_{1} (\cdot,\tau)\big] - \Gamma\big[\vphi_{2} (\cdot,\tau)\big]\big\}.
$$
Then, using \eqref{eq:ILip} as well as the bound $\| e^{-t\hat{D}}\|\leq 1$, we obtain
\bes
\vert| \mathcal T\vphi_1 (\cdot,\tau)- \mathcal T\vphi_2 (\cdot,\tau) \vert|_{\infty} \leq 2 \int_0^t d\tau \vert| \vphi_1 (\cdot,\tau)- \vphi_2 (\cdot,\tau) \vert|_{\infty} \leq 2 d_T(\vphi_1,\vphi_2) \big(1-e^{-T}\big)
\ees
whence 
$$d_T(\mathcal T\vphi_1,\mathcal T\vphi_2) \leq 2  \big(1-e^{-T}\big) d_T(\vphi_1,\vphi_2).$$
Thus, the operator $\mathcal T$ is contractive if $T<\log(2)$. 

Given that $\vphi (\cdot,\tau)\in \Phi$ for any $t\in [0,T)$ and therefore $\| \vphi (\cdot,\tau)\|_{\infty}\leq 1$ the solution can be extended to arbitrary values of $T\in (0,\infty)$ using a standard extension argument. 

\end{proof}

\begin{remark}
We observe that the proof of the previous theorem is rather standard. We notice that results which imply Theorem \ref{thm:1} can be found in \cite{CercArchive} and  in \cite{JNV} although these papers do not use the Fourier transform approach.
\end{remark}

\begin{remark}\label{eq:C1sol}
Notice that if $\vphi\in C([0,\infty);\Phi)$ is a solution to \eqref{eq:IntBoltFou}  then $\vphi\in C^1((0,\infty); \Phi)$. A standard computation shows that $\vphi$ solves the PDE problem  \eqref{eq:incdtFourier}, \eqref{eq:BoltFourier2}.  
In the rest of the paper we will not make any distinction between solutions to the integral equation \eqref{eq:IntBoltFou} and solutions to    \eqref{eq:incdtFourier}, \eqref{eq:BoltFourier2}.
\end{remark}
As we have seen, the classical Lipschitz property for norms (cf.\eqref{eq:ILip}) is sufficient to prove existence and uniqueness of solutions to \eqref{eq:IntBoltFou}. We will now show that the stronger pointwise Lipschitz property \eqref{eq:est} in Lemma \ref{lm:bdI} is useful to compare two solutions with different initial data. 

It will be convenient in the following to introduce the semigroup notation.  Let $C_b(\R^{d};\C)$  be the space of continuous and bounded functions with complex value. Given a function $y_0\in C_b(\R^{d};\C)$ we denote as $\left[\exp\big[\big(\mathcal{L}-{I}- \hat{D} \big)\ t\big] y_0\right](\cdot)$ the unique solution $y\in C([0,\infty); C(\R^{d};\C))$ of the equation
\be \label{eq:sgreq}
y(\cdot,t)=\left[E(t)y_0\right](\cdot)+\int_0^t E(t-\tau) \mathcal{L}\big(y(\cdot,\tau)\big)d\tau
\ee
where $E(t)$ is defined as in \eqref{eq:defE} and we recall that $\hat{D}=\big(Ak\big)\cdot \partial_{k}$. Notice that \eqref{eq:sgreq} represents the mild formulation of the equation
\bes
\partial_t y(\cdot,t)+ \big(Ak\big)\cdot \partial_{k} y(\cdot,t)=\mathcal{L}y(\cdot,t)-y(\cdot,t), \quad y(\cdot ,0)= y_{0}(\cdot) \ .
\ees
The existence and uniqueness of solutions to \eqref{eq:sgreq} can be obtained using a standard fixed point argument along the lines of the proof of Theorem \ref{thm:1}.

The following Theorem allows to estimate the difference of two solutions to  \eqref{eq:incdtFourier}, \eqref{eq:BoltFourier2} in terms of the semigroup associated to the operator $\mathcal{L}-{I}- \hat{D}$  where $\mathcal{L}$ is as in \eqref{eq:linBolzF} and $\hat{D}$ is as in \eqref{eq:notshort}.
\begin{theorem} \label{thm:2}
Let be $\vphi_{1}, \vphi_{2}\in C([0,\infty);\Phi)$ safisfying \eqref{eq:IntBoltFou} with initial data $\vphi_{1,0}, \vphi_{2,0} \in \Phi$. 
Then the following inequality holds for all $t>0$ and $k\in \R^d$
\be \label{eq:est2th2}
\vert \vphi_{1}(k,t)-\vphi_{2}(k,t)\vert \leq \left[ \exp\big[\big(\mathcal{L}-{I}- \hat{D} \big) \ t \big] \ \vert \vphi_{1,0}-\vphi_{2,0}\vert \right] (k),
\ee
with $\hat{D}$ as in \eqref{eq:notshort}.
\end{theorem}
\begin{proof}
We denote 
\bes
y(\cdot,t)=\vphi_{1}(\cdot,t)-\vphi_{2}(\cdot,t), \qquad y(\cdot,0)=y_0 (\cdot)= \vphi_{1,0}(\cdot)-\vphi_{2,0}(\cdot)
\ees
Then, using that $\vphi_{1} (\cdot, t), \vphi_{2} (\cdot, t)$ satisfy \eqref{eq:IntBoltFou}, we obtain 
\bes 
y(\cdot, t)= E(t)y_0(\cdot) + \int_0^t d\tau \ E(t-\tau) \left\{\Gamma\big[\vphi_{1} (\cdot,\tau)\big]-\Gamma\big[\vphi_{2} (\cdot,\tau)\big]\right\}.
\ees
Hence, for any fixed $k\in\R^d$, we obtain 
\begin{align*}
\vert y(\cdot,t) \vert \leq E(t)\vert y_0(\cdot)\vert +  \int_0^t d\tau \vert E(t-\tau) \left\{\Gamma\big[\vphi_{1} (\cdot,\tau)\big]-\Gamma\big[\vphi_{2} (\cdot,\tau)\big]\right\} \vert 
\end{align*}
where  $E(t)$ is as in \eqref{eq:defE}. 

We now use Lemma \ref{lm:bdS} with $S= \exp \big(-t \hat{D} \big)$ and obtain 
\be \label{eq:bdy}
\vert y(\cdot,t) \vert \leq E(t)\vert y_0(\cdot)\vert +  \int_0^t d\tau  E(t-\tau) \mathcal{L} \vert y(\cdot,\tau) \vert ,
\ee
where $\mathcal{L}$ is given as in \eqref{eq:linBolzF}. 

Let us consider the Cauchy problem for the linear equation
\be \label{eq:u}
\partial_t u+ \big(Ak\big)\cdot \partial_{k} u+u=\mathcal{L}u, \quad u(\cdot ,0)=\vert y_{0}(\cdot )\vert.
\ee
Then $u(k,t)$ satisfies the integral equation 
\be \label{eq:u_int}
u(\cdot ,t)=E(t)\vert y_0(\cdot )\vert +  \int_0^t d\tau  E(t-\tau) \mathcal{L}u(\cdot ,\tau).
\ee
A standard comparison argument, applied to $u(\cdot ,t)$ and $\vert y(\cdot ,t) \vert $ satisfying the inequality \eqref{eq:bdy}, shows that 
\be \label{eq:comparison}
\vert y(\cdot ,t) \vert \leq u(\cdot ,t)= \exp\big[\big(\mathcal{L}- {I}- \hat{D} \big) \ t\big]\vert y_0(\cdot )\vert.
\ee
Thus, using the definition of $y(\cdot ,t)$ and $y_0(t)$, the theorem follows.
\end{proof}

In the next section we consider some applications of Theorem \ref{thm:2}.  More precisely, the way in which we will apply this Theorem is the following. Suppose that $f_1, f_2 $ are two solutions of \eqref{eq:GenBolt} having identical moments up to a given order $N$. Then, the corresponding characteristic functions $\vphi_1, \vphi_2$ will be tangent up to order $N$ at $k=0$. Hence, the estimate \eqref{eq:est2th2} will imply that the characteristic functions $\vphi_1, \vphi_2$ become close to each other, as $t\to \infty$ if $N$ is sufficiently large.

\section{Long-time asymptotics}  \label{sec:5}

In order to study the long-time asymptotics we would relate the moments of the probability distribution $f(v)$ with the Taylor series of the corresponding characteristic function $\vphi(k)$.
To this end we introduce some notation for multi-indexes. We will write $\N_{\star}=\{0,1,2,\dots \}$. A multi-index $\alpha$ will be an element of $\N_{\star}^d$, i.e. $\alpha=(\alpha_1 \dots \alpha_d)$ with $\alpha_j\in\N_{\star}$. We then define $|\alpha|=\alpha_1+\dots+\alpha_d$ and $\alpha !=\prod_{j=1}^d (\alpha_j)!$. Given any $\xi\in\R^d$ with the form $\xi=(\xi_1, \dots, \xi_d)$, $\xi_j \in\R$ we define $\xi^{\alpha}=\xi_1^{\alpha_1}\dots \xi_d^{\alpha_d}=\prod_{j=1}^d (\xi_{j})^\alpha_j $. Using this notation we have the following formula
\be\label{eq:expseries}
e^{-i k\cdot v}=\sum_{\alpha\in\N_{\star}^d} \frac{(-1)^{|\alpha}}{(\alpha )!} k^{\alpha }v^{\alpha },\quad v\in\R^d, \ k\in \R^d.
\ee

Suppose that $f\in\mathcal{P}_{+}(\R^d)$. Then, using the definition of $\vphi$ given in \eqref{eq:phif} as well as \eqref{eq:expseries} we obtain formally that 
\be \label{eq:series}
\vphi(k)=\sum_{\alpha\in\N_{\star}^d} \frac{(-i)^{|\alpha|}}{(\alpha)!}m_{\alpha} k^{\alpha} \quad \text{where}\quad m_{\alpha}=\int_{\R^d} dv f(v) \ v^{\alpha}.
\ee
Notice that \eqref{eq:series} holds if the moments of the measure $f$ are not too large for large $|\alpha|$. For instance, if the probability measure $f$ is supported in the set $\{|v|\leq A\}$ we will have $\vert m_{\alpha}\vert\leq A^{|\alpha|}$ and the series \eqref{eq:series} would then be convergent for any value of $k\in\R^d$.  We will not require in this paper such a strong condition. We just remark that if $f$ has moments $m_{\alpha}$ with $|\alpha|\leq N$ then the function $\vphi\in C^{N-1}(\R^d)$. Moreover, if the moments $m^1_{\alpha}, \ m^2_{\alpha}$ of two measures $f_1,\ f_2\in \mathcal{P}_{+}(\R^d)$ are the same for $|\alpha|\leq N$  then the corresponding characteristic functions $\vphi_1, \ \vphi_2$ have finite moments up to order $N$ and take the same values for $|\alpha|\leq N-1$. Then, $\vphi_1, \ \vphi_2$ have the same Taylor series up to order $N-1$. 

Notice that since $f \in \mathcal{P}_{+}(\R^d)$ we have that the zero-th order moment $m_O=1$ where $O=(0,0,\dots,0)$.

\begin{lemma}\label{lem:estfi1fi2}
Suppose that $f_1, f_2\in \mathcal{P}_{+}(\R^d)$ satisfy the conditions 
\be \label{eq:moments}
\int_{\R^d} dv f_j(v) \vert  v\vert^N<\infty,\quad j=1,2. 
\ee
Suppose that in addition the moments $m^1_{\alpha}, \ m^2_{\alpha}$ associated to $f_1,\ f_2$ respectively, satisfy $ m^1_{\alpha}=m^2_{\alpha}$ for any $|\alpha|\leq N$.  
Let $\vphi_{1}, \ \vphi_2\in \Phi$ be the corresponding characteristic functions. Then, there exists a function $\omega=\omega(k)$, $\omega:\R^d\rightarrow \R_{+}$ satisfying 
\be\label{def:omegak}
\lim_{k\to 0 }\omega(k)=0 \quad \text{and}\quad \omega(k)\leq C_1,
\ee
such that the following estimate holds:
\be \label{eq:estfi1fi2}
\vert \vphi^{1}(k)-\vphi^{2}(k) \vert \leq 
\omega(k)\vert k\vert^N,\quad k\in \R^d 
\ee
where the function $\omega$ depends on the measures $f_1, f_2$, on the order $N$ and on the dimension $d$. The constants $C_1$ depends on the moments \eqref{eq:moments}, on $N$ and $d$. 
\end{lemma}

\begin{remark}
As it will be seen in the proof, the rate of convergence of the function $\omega(\cdot)$ as $k\to 0$ (cf. \eqref{def:omegak}) depends on the rate of convergence of the integrals $\int_{\{\vert v\vert\leq R\}} \vert v\vert^N f_j(v) dv$ to the limit value $\int_{\R^d} \vert v\vert^N f_j(v) dv$ as $R\to \infty$. In order to obtain uniform estimates for $\omega(\cdot)$ it will be enough to assume, for instance, that $\int_{\R^d} \vert v\vert^{N+\delta} f_j(v) dv<\infty$ for some $\delta>0$.
\end{remark}

\begin{proof}
We use the following inequality which is a consequence of the Taylor series for the exponential function:
\be \label{eq:SN}
 \left\vert e^{-i (k\cdot v)} - \sum_{\vert \alpha \vert \leq N} \frac{(-i)^{\vert \alpha \vert }}{( \alpha )!} k^{\alpha} v^{\alpha} \right\vert 
 \leq C\min\{\vert k\vert ^{N+1}\vert v\vert ^{N+1}, \vert k\vert ^{N}\vert v\vert ^{N}\}
\ee
where $C$ is a constant that depends on $d$ and $N$. Notice that \eqref{eq:moments} implies that there exists a function $\delta=\delta(R)$, $\delta:\R^{+}\to\R^{+} $ satisfying $\lim_{R\to \infty} \delta(R)=0$, such that 
\be\label{eq:deltaR}
\max_{j=1,2} \left( \int_{\{\vert v\vert > R\}} \vert v\vert^N f_j(v) dv \right) \leq \delta(R) .
\ee 
Using \eqref{eq:phif} as well as \eqref{eq:SN}, and splitting the domain of integration in the remainder as $\{\vert v\vert \leq R \} \cup \{\vert v\vert > R\}$ we obtain
 \begin{align} \label{eq:SN1}
 \left\vert \vphi_j(k)- \sum_{\vert \alpha \vert \leq N} \frac{(-i)^{\vert \alpha \vert }}{( \alpha )!}  m^j_{\alpha} k^{\alpha}\right\vert 
&\leq C \int_{\{\vert v\vert \leq R\}}\vert k\vert ^{N+1}\vert v\vert ^{N+1} f_j(v) dv +C \int_{\{\vert v\vert > R\}}\vert k\vert ^{N}\vert v\vert ^{N} f_j(v) dv \nonumber \\&
\leq  C \vert k\vert ^{N+1} R \int_{\R^d} \vert v\vert^{N} f_j(v) dv +C \vert k\vert ^{N} \delta(R)
\end{align}
We now remark that for any $\ep>0$  arbitrarily small, taking $R$ sufficiently large and then $k$ sufficiently small, we can estimate the right hand side of \eqref{eq:SN1} as $\ep  \vert k\vert ^{N} $. This proves the existence of the function $\omega(k)$ with the properties stated in \eqref{def:omegak} such that: 
\bes
 \left\vert \vphi_j(k)- \sum_{\vert \alpha \vert \leq N} \frac{(-i)^{\vert \alpha \vert }}{( \alpha )!}  m^j_{\alpha} k^{\alpha}\right\vert 
\leq \frac{\omega(k)}{2} \vert k\vert ^{N} \quad j=1,2.
\ees
Therefore, using this formula, as well as the fact that $m^1_{\alpha}=m^2_{\alpha}$ for $|\alpha|\leq N$ we can estimate the difference $\vert \vphi^{1}(k)-\vphi^{2}(k) \vert $ as in \eqref{eq:estfi1fi2} and the Lemma follows. 
\end{proof}

We will denote as $C_p(\R^{d};\C)$ the normed space of continuous functions such that 
$\vert|\psi\vert|_p:=\sup_{k\in\R^d} \frac{\vert \psi(k)\vert}{1+\vert k\vert^p}<\infty$ with $0<p<\infty$. 

\begin{lemma}\label{lem:estEL}
Let the operators $\mathcal{L}$ and $E(t)$ be defined as in \eqref{eq:linBolzF},  \eqref{eq:defE} respectively.
Then, the following estimates hold:
\begin{align}  
& \vert| \mathcal{L}\psi \vert |_{p} \leq 2 \vert| \psi \vert|_{p}\ , \label{est:L} \\
&\vert| E(t)\psi \vert|_{p} \leq e^{-(1-p\vert | A \vert |)t }\vert| \psi\vert|_{p},\quad t\geq 0 \ . \label{est:E}
\end{align}
\end{lemma}
\begin{proof}
In order to prove \eqref{est:L} we notice that \eqref{eq:linBolzF} implies 
\bes 
\vert \mathcal{L}\big[\vphi\big](k)\vert \leq \int_{S^{d-1}}dn \ g\left(  \hat{k} \cdot n \right) \vert| \psi\vert|_{p}\left(2+\vert k_{+}\vert+ \vert k_{-}\vert\right)
\ees
Using than that $\vert k_{\pm}\vert \leq \frac 1 2 \vert k\vert (1+\vert n\vert )\leq \vert k \vert$ we obtain 
\bes 
\vert \mathcal{L}\big[\vphi\big] (k)\vert \leq 2 \vert| \psi\vert|_{p}\left(1+ \vert k \vert^p \right)
\ees
whence  \eqref{est:L}  follows.  
Estimate \eqref{est:E} follows from the following chain of inequalities:
\begin{align*}
\vert E(t)\psi \vert  &\leq e^{-t} \vert \psi\left(e^{-At}k\right)\vert \leq e^{-t} \vert| \psi\vert|_{p}\left(1+\vert e^{-At}k \vert^p  \right)\\&
\leq e^{-t} \vert| \psi\vert|_{p} \left(1+e^{p\vert|A\vert| t}\vert k \vert^p  \right) \leq  e^{-t+p\vert|A\vert| t} \vert| \psi\vert|_{p}\left(1+\vert k \vert^p  \right) \ .
\end{align*}
\end{proof}
\smallskip

We collect in the following proposition some properties of the semigroup $\exp\big[\big(\mathcal{L}-I - \hat{D} \big) \ t\big]$. 

\begin{proposition}\label{prop:supersol}
The following results hold.
\begin{itemize}
\item[(i)] For any $y_0\in C_p(\R^{d};\C)$ there exists a unique solution $y\in C([0,\infty); C_p(\R^{d};\C))$ to \eqref{eq:sgreq}. We will denote this solution as 
\bes
y(\cdot, t)=\exp\big[\big(\mathcal{L}-I - \hat{D} \big) \ t\big]y_0(\cdot).
\ees
\item[(ii)] Suppose that $y_0, u_0 \in C_p(\R^{d};\C)$ with $y_0(k), u_0(k)\in\R_{+}$ for any $k\in\R^d$ and $y_0(k) \leq u_0(k)$.
Then 
\be \label{eq:mon1}
0\leq \exp\big[\big(\mathcal{L}-I - \hat{D} \big) \ t\big]y_0(\cdot) \leq \exp\big[\big(\mathcal{L}-I - \hat{D} \big) \ t\big]u_0(\cdot).
\ee
\item[(iii)] Let $y_0\in C_p(\R^{d};\C)$, $y_0(k), u_0(k)\in\R_{+}$ and suppose that $u\in C([0,\infty); C_p(\R^{d};\C))$ is a nonnegative function satisfying
\be\label{eq:ubound}
u(\cdot,t)\geq (E(t)y_0)(\cdot)+\int_0^t d\tau E(t-\tau) \mathcal{L}(u(\cdot, \tau)) \quad \text{for any}\quad t\geq 0. 
\ee
Then 
\be\label{eq:supersol}
\exp\big[\big(\mathcal{L}-I - \hat{D} \big) \ t\big]y_0(\cdot)\leq u(\cdot,t).
\ee
\end{itemize}
\end{proposition}
\begin{remark}
Notice that for any function $\psi\in C_p(\R^{d};\C)$  we have $\big(E(t)\psi\big)(k)=e^{-t}\psi(e^{-At}k).$ Thus $E(t)$ maps $C_p(\R^{d};\C)$ into itself for any $t\geq 0$ and therefore equation \eqref{eq:sgreq} is meaningful for $y\in C([0,\infty);C_p(\R^{d};\C))$.
\end{remark}

\begin{proof}
In order to prove item $(i)$, namely the existence of a solution $y$  satisfying \eqref{eq:sgreq}, we define the following iterative sequence 
\be \label{eq:iterzn}
z_0(\cdot,t)=0,\quad z_{n+1}(\cdot,t)= E(t)y_0(\cdot)  + \int_{0}^{t} d\tau \ E(t-\tau) \mathcal{L} z_{n}(\cdot,\tau ). 
\ee 
Then we consider the p-norm of the difference $z_{n+1}-z_{n}$ for $n\geq 1$, using the equation above. Using Lemma \ref{lem:estEL} and setting $c\leq 1-p\vert | A \vert | $ we obtain
\bes 
\|z_{n+1}(t)- z_{n}(t)\|_p  \leq 2\int_0^t  e^{-c(t-\tau) }\vert| z_{n+1}(\tau)-z_{n}(\tau)\vert|_{p} .
\ees
Notice that we can assume without loss of generality that $c\neq 0$ (it may be negative). Then, for a given $T> 0$, 
$$\sup_{0\leq t\leq T}\|z_{n+1}(t)- z_{n}(t)\|_p \leq \frac {2}{\vert c\vert } \vert 1- e^{-cT }\vert \sup_{0\leq t\leq T}\|z_{n}(t)- z_{n-1}(t)\|_p,\quad n\geq 2. $$
On the other hand, since $z_0(\cdot,t)=0$,  we have
$$\sup_{0\leq t\leq T}\|z_{1}(t)- z_{0}(t)\|_p \leq \max\{1, e^{-cT }\} \vert| y_0\vert|_{p}.$$
Then, if $T>0$ is sufficiently small we obtain that $\{z_{n}\}_{n\geq 0} $ is a Cauchy sequence in $C([0,T); C_p(\R^{d};\C))$.  Thus there exists $y\in C([0,T); C_p(\R^{d};\C)) $ such that $\lim_{n\to\infty} z_{n}= y$ in $C([0,T); C_p(\R^{d};\C)) $. 
Taking the limit as $n\to \infty$ in \eqref{eq:iterzn} it follows that $y$ satisfies \eqref{eq:sgreq}. Due to the linearity of the problem  \eqref{eq:sgreq} we can use the standard extension argument to obtain $T$ arbitrary large. Therefore the existence part of item (i) follows. 
To prove uniqueness we suppose to have two different solutions $y_1, y_2 \in C([0,T); C_p(\R^{d};\C))$. Their difference satisfies
\bes 
y_1(\cdot,t)- y_2(\cdot,t)=E(t)y_0(\cdot)  + \int_{0}^{t} d\tau \ E(t-\tau) \mathcal{L}\left( y_1- y_2\right)(\cdot,\tau ).
\ees
Then, Lemma \ref{lem:estEL} implies
\bes 
\| y_1(t)- y_2(t)\|_p\leq 2 \int_{0}^{t} d\tau \ e^{-c(t-\tau)}  \| y_1(\tau)- y_2(\tau)\|_p
\ees
whence $y_1(\cdot,t)=y_2(\cdot,t)$ by using Gronwall Lemma.  

We now prove item $(iii)$.  Suppose that $u(\cdot,t)$ is such that the inequality \eqref{eq:ubound} holds. We can then prove that the sequence of functions $\{z_{n}\}$ defined in \eqref{eq:iterzn} satisfies 
\be \label{est:znu}
z_{n}(\cdot,t)\leq u(\cdot,t) \quad \text{ for any}\quad n\geq 0.
\ee 
This claim follows by induction. Indeed, for $n=0$ we have $z_{0}(\cdot,t)=0\leq u(\cdot,t)$. Suppose now that $n\geq 0$ and that $z_{n}(\cdot,t)\leq u(\cdot,t)$. Then 
\bes
z_{n+1}(\cdot,t)\leq E(t)y_0(\cdot)  + \int_{0}^{t} d\tau \ E(t-\tau) \mathcal{L} u(\cdot,\tau ) \leq u(\cdot,\tau ),
\ees 
where in the first inequality we used the definition \eqref{eq:iterzn} and the induction assumption, while in the second inequality we used \eqref{eq:ubound}. Hence, \eqref{est:znu} follows whence $y=\lim_{n\to \infty} z_n \leq u$ and this yields item $(iii)$. 

We now prove item $(ii)$. To this end we notice that the definition of $\exp\big[\big(\mathcal{L}-I - \hat{D} \big) \ t\big]u_0(\cdot)=u(\cdot,t)$ implies 
\bes 
u(\cdot,t)=\left[E(t)u_0\right](\cdot)+\int_0^t E(t-\tau) \mathcal{L}\big(u(\cdot,\tau)\big)d\tau. 
\ees 
Using then $E(t)u_0 \geq E(t)y_0$ it follows that the inequality \eqref{eq:ubound} holds. Whence, \eqref{eq:mon1} follows from item $(iii)$ and the proposition is proved.
\end{proof}

\begin{lemma}\label{lem:bdup}
For any $p>0$ and any matrix $A\in M_{d\times d}(\R)$  we define the function 
\be \label{eq:defup}
u_p(k,t)=\vert k \vert^p \exp\big[-\big(\lambda(p)-p\vert | A\vert |\big)t\big],\quad k\in\R^d,\; t\geq 0,
\ee  
where 
\be \label{eq:lambdap}
\lambda(p)=\int_{S^{d-1}} dn \ g(\omega \cdot n) \left[1- \left(\frac{1+(\omega\cdot n)}{2}\right)^{\frac p 2}-\left(\frac{1-(\omega\cdot n)}{2}\right)^{\frac p 2} \right].
\ee
Then, $u_p(k,t)\in C([0,T); C_{p}(\R^d;\C))$ and satisfies the inequality (cf. \eqref{eq:ubound})
\be \label{eq:bdup}
 u_{p}(k,t)\geq E(t)\vert k \vert^p+\int_0^t d\tau \left( E(t-\tau) \mathcal{L}(u_p(\cdot, \tau))\right)(k) \quad \text{for any}\quad t\geq 0.
\ee
\end{lemma}

\begin{remark}
We notice that $\lambda(p)$  is the eigenvalue of the linearized collision operator $-\big(\mathcal{L}-I\big)$ (cf.~\eqref{eq:linBolzF}) that corresponds to the eigenfunction $\vert k \vert^p$ and the norm $\|A\|$ is defined as in \eqref{eq:normA2}. More precisely, 
\be \label{eq:Lkp}
\mathcal{L} \vert k\vert^p= \big( 1-\lambda(p)\big) \vert k\vert^p,\quad p>0.
\ee
We also remark that if $A=-\vert | A\vert | I$, where $I$ is the identity matrix, than the inequality \eqref{eq:bdup} becomes an identity. 
\end{remark}

\begin{proof}
To prove \eqref{eq:bdup} we argue as follows. 
We compute the right hand side of \eqref{eq:bdup} and we have
\begin{align*}
J:&= E(t)\vert k \vert^p+\int_0^t d\tau \left( E(t-\tau) \mathcal{L}(u_p(\cdot, \tau))\right)(k) \\&=   
E(t)\vert k \vert^p+ \big( 1-\lambda(p)\big) \int_0^t d\tau  e^{-(\lambda(p)-p\vert | A\vert |)\tau} E(t-\tau) \vert k\vert^p
\end{align*}
where we used \eqref{eq:Lkp}. 
Note that $\big( 1-\lambda(p)\big) >0$. 
 Using also that $E(t)\vert k \vert^p=e^{-t} \vert e^{-At} k\vert^p \leq e^{-t}e^{p\vert | A\vert| t}\vert k\vert^p$ we obtain
\begin{align*}
 J& \leq e^{-t}e^{p\vert | A\vert| t}\vert k\vert^p+  \big( 1-\lambda(p)\big)\vert k\vert^p \int_0^t d\tau e^{-(t-\tau)} e^{p\vert | A\vert| (t-\tau)}e^{-(\lambda(p)-p\vert | A\vert |)\tau} \\&
\leq  e^{-(\lambda(p)-p\vert | A\vert |)t}\vert k\vert^p=u_p(k,t)
\end{align*} 
and the result follows. 
\end{proof}

The fact that the equality sign  in \eqref{eq:bdup} holds for $A=-\vert | A\vert | I$ follows from the fact that for this particular choice of the matrix $A$ we have $E(t)\vert k \vert^p=e^{-t}e^{p\vert | A\vert| t } \vert  k\vert^p$. 

\medskip

We can now state the following fundamental result which is the main result of this section.
\begin{theorem} \label{thm:3}
Suppose that we consider two solutions $\vphi_{1}(k,t)$, $\vphi_{2}(k,t)$ to \eqref{eq:IntBoltFou}, $\vphi_{1}, \vphi_{2}\in C([0,T); \Phi) $ such that 
\be
\label{eq:estfi1fi2time0}
\vert \vphi_{1}(k,0)-\vphi_{2}(k,0) \vert \leq B\vert k\vert^p,\quad
\ee
for any $k\in \R^d$ and some $p>0$, $B>0$. Then the following inequality holds for all $t\geq 0$ and $k\in\R^d$:
\be\label{eq:estfi1fi2timet}
\vert \vphi_{1}(k,t)-\vphi_{2}(k,t) \vert \leq B\vert k\vert^p \exp\big[-\big(\lambda(p)-p\vert | A\vert |\big)t\big].
\ee
Moreover,  the function $\lambda(p)$ given in \eqref{eq:lambdap} is an increasing function for $p\in [0,\infty)$ and it satisfies $\lambda(p)>-1=\lambda(0)$ for $p>0$ and $\lambda(p)>0=\lambda(2)$ for $p>2$.
\end{theorem}

\begin{proof} 
Combining the inequality \eqref{eq:est2th2} in Theorem \ref{thm:2} with item $(ii)$ in Proposition \ref{prop:supersol} as well as \eqref{eq:estfi1fi2time0} we obtain 
\bes 
\vert \vphi_{1}(k,t)-\vphi_{2}(k,t) \vert \leq B  \exp( (\mathcal{L}-I-\hat{D})t)  (\vert k\vert^p).
\ees 
Using now item $(iii)$ in Proposition \ref{prop:supersol} as well as Lemma \ref{lem:bdup} we obtain $e^{(\mathcal{L}-I-\hat{D})t}\vert k\vert^p \leq u_p(k,t)$ whence the theorem follows.

The monotonicity of the function $\lambda(p)$ is obtained differentiating  \eqref{eq:lambdap}, namely
\bes
\frac{d \lambda}{dp}=-\frac 1 2 \int_{S^{d-1}} dn \ g(\omega \cdot n) \left[\left(\Omega_{+}\right)^{\frac p 2} \log\left(\Omega_{+}\right) +\left(\Omega_{-}\right)^{\frac p 2}\log \left(\Omega_{-}\right) \right] > 0,
\ees
where $\Omega_{\pm}= \frac{1\pm (\omega\cdot n)}{2}\in [0,1]$ and then $\log(\Omega_{\pm})\leq 0$.   Since $\lambda(0)=-1$  and $\lambda(2)=0$ the theorem follows. 
\end{proof}

\bigskip 
We now briefly discuss  the meaning of Theorem \ref{thm:3}. The inequality \eqref{eq:estfi1fi2timet} implies 
\bes
\lim_{t\to\infty} \vert \vphi_{1}(k,t)-\vphi_{2}(k,t)  \vert =0,\quad k\in\R^d
\ees
provided that $p>2$ and $\vert | A \vert | < \frac 1 p \lambda(p).$

\medskip

 In particular, if $f_1, f_2 \in C([0,T); \mathcal{P}_{+}(\R^d))$ are two solutions of \eqref{eq:GenBolt} with initial data $f_{1,0}, f_
 {2,0}$,  having identical moments at time zero, i.e.~$m^1_{\alpha}(0)=m^2_{\alpha}(0)$ with $|\alpha|\leq N$, $N\geq 3$,  then, estimating the right hand side of \eqref{eq:estfi1fi2}  in Lemma \ref{lem:estfi1fi2} by $C \vert k\vert^N$ and using Theorem \ref{thm:3}, we obtain that $f_1(\cdot,t), f_2(\cdot,t)$ will be close as $t\to \infty$ in the weak topology of probability measures. This is due to the classical result (see for instance \cite{Fe2}) that relates the pointwise convergence of characteristic functions with the weak convergence of probability measures. Notice that in the simplest case $A=0$ the stationary solutions of \eqref{eq:GenBolt} are given by Maxwellian distributions $f_{st}(\cdot)$.
Then, if $f_{2}(\cdot,t)= f(\cdot,t) $ is any other solution of \eqref{eq:GenBolt} having the same moments up to order $|\alpha|\leq 3$ as $f_{st}(\cdot)$ at time zero, it then follows, taking $f_{1}(\cdot,t)=f_{st}(\cdot)$, that   
\be \label{eq:asym}
\lim_{t\to\infty} \vphi(k,t)=\vphi_{st}(v), \quad k\in\R^d ,
\ee 
where $\vphi_{st}(\cdot)$ is the characteristic function associated to  $f_{st}(\cdot)$. Hence $f(\cdot,t) $ converges to $f_{st}(\cdot)$ as $t\to \infty$ in the weak topology of measures. 

We can apply a similar strategy in the case $A\neq 0$ although in this case we need to take into account that in general there are no steady states of \eqref{eq:GenBolt}. On the contrary, we will show that the long-time asymptotics of the solutions to \eqref{eq:GenBolt} is given by a family of self-similar solutions.  On the other hand, some additional analysis is needed to study the long-time behaviour of the second order moments. This is due to the fact that  we can expect the long-time behaviour of the solutions to \eqref{eq:GenBolt} to be characterized only by  moments of order zero and one, as well as the kinetic energy $\int_{\R^d} f(v,t) \vert v\vert^2 dv$.

\section{Generalities about self-similar and stationary solutions}  \label{sec:6}

In this section we discuss in a formal way the expected long-time asymptotics for the solutions to the integral equation \eqref{eq:IntBoltFou} (cf. also \eqref{eq:BoltFourier2}). 
 It is illuminating to begin describing first the situation for $A=0$ since in this case \eqref{eq:BoltFourier2} is just the  classical Boltzmann equation in the Fourier representation.    

It is readily seen, using \eqref{eq:BoltCollFourier}, that \eqref{eq:BoltFourier2} with $A=0$ has the following family of steady states 
\be\label{eq:PsiAzero}
\Psi_{st}(k)=e^{-i U\cdot k}e^{-b |k|^2}, \quad k\in\R^d,\; U\in\R^d,\; 0\leq b<\infty. 
\ee
Note that we always use the normalization $\Psi(0)=1$. Actually, there is a natural proof, using the machinery developed in the previous sections, that these are the only steady states to \eqref{eq:BoltFourier2} in a suitable class of characteristic functions. More precisely we have the following
\begin{proposition}\label{prop:mom}
Let $\Psi\in\Phi$ such that 
\be \label{eq:ExMomPsi}
\vert \Psi(k)-\sum_{\vert \alpha \vert\leq 2}c_{\alpha}(k)^\alpha\vert \leq C\vert k\vert^p \quad \forall \; k\in\R^d, \; p>2,
\ee
for some $c_{\alpha}\in\C$ and $C>0.$ Suppose also that $\Psi$ is a steady state to \eqref{eq:BoltFourier2}  with $A=0$, i.e. $ \Psi(\cdot)=\Gamma[ \Psi] (\cdot)$. Then there exists a vector $U\in \R^d$ and a real constant $b\geq0$ such that 
$\Psi$ is as in \eqref{eq:PsiAzero}.
\end{proposition}

\begin{proof}
We first observe that results equivalent to this one could be proved using for instance the $H-$Theorem. However, we will include here a proof based on the Fourier Transform method for completeness. 

Given that $\Psi\in\Phi$ we have that $\Psi(0)=1$.  
Using Taylor series as well as \eqref{eq:ExMomPsi} there exists a (unique) vector $U\in \R^d$ such that the function $e^{i U\cdot k}\Psi(k)$ satisfies
\bes
\vert e^{i U\cdot k}\Psi(k) -\big(1+\sum_{\vert \alpha \vert= 2}\tilde{c}_{\alpha}(k)^\alpha\big)\vert \leq C\vert k\vert^p \quad \forall \; k\in\R^d, \; p>2
\ees
with $\tilde{c}_{\alpha}\in\C$. Moreover, using \eqref{eq:BoltCollFourier} if follows that if $A=0$, $\Psi(k)$ is a steady state of  \eqref{eq:BoltFourier2} if and only if $e^{i U\cdot k}\Psi(k) $ is a steady state of \eqref{eq:BoltFourier2}. Therefore, we can assume without loss of generality that ${c}_{\alpha}=0$ for $\vert \alpha \vert=1$ in \eqref{eq:ExMomPsi}. We can then rewrite \eqref{eq:ExMomPsi} as 
\bes
\Psi(k)= 1+ \sum_{j,\ell} b_{j,\ell} k_j k_{\ell}+O(\vert k\vert^p)
\ees
where $B=(b_{j,k})_{j,k}$ is a symmetric matrix with real coefficients.  
Plugging now this expansion in the equation for the steady states to \eqref{eq:BoltFourier2} we obtain
\begin{align*}
B: (k\otimes k)&=\int_{S^{d-1}} dn\ g(\hat{k}\cdot n ) \left[B: (k_{+}\otimes k_{+})+B: (k_{-}\otimes k_{-})\right] \\&
= \frac{1}{2} B: (k\otimes k)+\frac{\vert k\vert^2}{2}  \int_{S^{d-1}} dn\ g(\hat{k}\cdot n )B: (n\otimes n)
\end{align*} 
where we used the standard tensor notation $B: (v\otimes v)=b_{j,k} v_{j}v_{k}$ for any $v\in\R^d$.
Hence
\bes
B: (k\otimes k)= \vert k\vert^2 \int_{S^{d-1}} dn\ g(\hat{k}\cdot n ) B: (n\otimes n)= \frac{  \vert k\vert^2  }{d} \Tr(B),\quad k\in\R^d
\ees
where we used that $\int_{S^{d-1}} dn\ g(\hat{k}\cdot n ) (n\otimes n)= \frac{1}{d} I$. This implies that $B$ is  a diagonal matrix, i.e. $B=-b I$, $b\in\R$. 
Thus
\be\label{eq:PsiTaylorBis}
\vert \Psi(k)-(1-b \vert k \vert ^2 )\vert\leq C \vert k\vert^p,\quad b\geq 0, \; p>2,\;k\in\R^d .
\ee
Using then Theorem \ref{thm:3} with $\vphi_1=\Psi$ and $\vphi_2=\Psi_{st}$ as in the right hand side of \eqref{eq:PsiAzero} (with $U=0$ and the same value of $C$) we obtain that  
\bes 
\vert \Psi(k)- \Psi_{st}(k)\vert \leq C\vert k\vert^p \exp\big(-\lambda(p) t\big), \quad \lambda(p)>0 .
\ees
Notice that \eqref{eq:estfi1fi2time0} holds due to \eqref{eq:PsiTaylorBis} and the Taylor series at $k=0$ for $\Psi_{st}$ in \eqref{eq:PsiAzero}. 
Taking the limit as $t\to \infty$ we obtain that $\Psi (k)\equiv \Psi_{st}(k)$ for any $k\in\R^d$ and the result follows. 
\end{proof}

\begin{remark}
Notice that the stability of solutions for the time-dependent problem \eqref{eq:BoltFourier2} with $A=0$ can be proved using similar arguments. However, since the case $A=0$ is a particular case of the results discussed in this paper we will not continue a separate analysis here. 
\end{remark}
 
 It turns out that in the case $A\neq 0$ the solutions to \eqref{eq:BoltFourier2} (or \eqref{eq:IntBoltFou})
do not behave in general as a steady state as $t\to \infty$. Actually, the stationary version of \eqref{eq:BoltFourier2} 
\be \label{eq:StatBoltFourier2}
\vphi + \left(Ak\right)\cdot \partial_k \vphi = \mathcal{I}^{+}(\vphi,\vphi) . 
\ee 
does not have any non-trivial solution for arbitrary matrices $A$.   
In particular, the proof of the fact that \eqref{eq:StatBoltFourier2} does not have any non-trivial  solution follows from the analysis of the equations for the second order derivatives of $\vphi$ (which corresponds to the second order moments of the associated probability measure $f$). 
Note that we are interested mainly in solutions with uniformly bounded second derivatives in $k$, which correspond to distribution functions having bounded second moment (energy).

We first recall a standard result that guarantees that the regularity of $\vphi_0(\cdot)$ at $k=0$ is propagated in time. 
\begin{proposition}\label{prop:regularity}
Let $N\in\N$ such that $N\geq 1$. Suppose that $\vphi_0\in\Phi$ and it satisfies
\bes
\vert \vphi_0(k)-\big(1+\sum_{|\alpha|\leq N} c_{\alpha} k^{\alpha}\big) \vert \leq C_0 \vert k\vert ^p,\quad N<p\leq N+1,
\ees
for some positive $C_0\in \R$.  Suppose that $\vphi \in C([0,\infty);\Phi)$ solves \eqref{eq:IntBoltFou}. Then we have 
\be 
\vert \vphi(k,t)-\big(1+\sum_{|\alpha|\leq N} b_{\alpha}(t) k^{\alpha}\big) \vert \leq C_0(t) \vert k\vert ^p,\quad N<p\leq N+1,
\ee
where $b_{\alpha}\in C[0,\infty)$ such that $b_{\alpha}(0)=c_{\alpha}$ for $\vert \alpha\vert\leq N$, and $C_0(\cdot)\in  C[0,\infty)$. 
\end{proposition}
The proof of this result is a straightforward adaptation of an analogous standard result for the classical Boltzmann equation (i.e. $A=0$). The main idea is to use the fact that the moments  $c_{\alpha}$ with $\vert \alpha\vert=\ell$ satisfy a linear ODE which has as right hand side the moments $c_{\alpha}$ with $\vert \alpha\vert\leq \ell-1$. Thus all the moments can be estimated for arbitrary long times using a Gronwall argument. The same Gronwall type of estimate can be applied to control the remainder $\vert k\vert^p$. Since the argument is classical we will not include more details in this paper.

\bigskip

We will justify now that the most general behaviour as $t\to\infty $ for the solutions of \eqref{eq:BoltFourier2} is given by
\be\label{eq:PssGen}
\vphi (k,t)=\exp\left(-i \big(e^{-t A^T}U\big)\cdot k\right)\Psi (ke^{\beta t}),\quad \beta \in \R,\; U\in\R^d.
\ee
Moreover, we will prove that there exists solutions to  \eqref{eq:BoltFourier2} having exactly the form on the right hand side of \eqref{eq:PssGen}. 

Suppose that $\vphi(\cdot,t)$ is a solution to \eqref{eq:BoltFourier2} that satisfies for any time \eqref{eq:ExMomPsi}. 
Then, it is easy to see that we can rewrite the function $\vphi$ as 
\bes
\vphi (k,t)=e^{-i w(t)\cdot k}\psi(k,t),\quad w(t)\in\R^d 
\ees
where
\bes
 \psi(k,t)=1+O(|k|^2)\quad \text{as}\quad |k|\to 0,\quad \forall \;t\geq 0. 
 \ees
Then 
$$\vphi (k,t)=(1-iw(t)\cdot k)+O(|k|^2)\quad \text{as}\quad |k|\to 0.$$ Plugging this expansion into \eqref{eq:BoltFourier2} and using that $\Gamma[\vphi]=\vphi=O(|k|^2)$ as $|k|\to 0$, we formally obtain 
\bes
\big(\partial_t w\big) \cdot k+\big( Ak\big)\cdot w=O(|k|^2) \; \text{as}\; |k|\to 0.
\ees
Since $k$ is arbitrary we then have
\bes 
\partial_t w+ A^T w=0
\ees
whence 
\bes
w(t)=e^{-t A^T}U,\quad \text{with}\quad w(0)=U\in\R^d.
\ees 
This implies that $\vphi (k,t)=\exp\left(-i \big(e^{-t A^T}U\big)\cdot k\right)\psi(k,t)$. 
Notice that the phase factor appearing on the right hand side of this equation corresponds to a displacement of the velocity. It is worth to remark that in the case of the classical Boltzmann equation ($A=0$) this phase is time-independent. This phase factor will be equal to one if the initial momentum of the distribution $U$ is zero.

In the original Boltzmann equation \eqref{eq:GenBolt} we can impose that the initial average velocity of the distribution of particles is zero using a suitable Galilean transformation as indicated in \cite{JNV3}. With the approach based on characteristic functions used in this paper, the natural approach to remove this phase term (which acts trivially on the collision operator) is to perform the change of variables $\psi (k,t)=\exp\left(i \big(e^{-t A^T}U\big)\cdot k\right)\vphi(k,t)$ and then $\psi (k,t)$ solves 
\bes 
\partial_t \psi  + \left(Ak\right)\cdot \partial_k \psi = \Gamma[\psi]-\psi 
\ees 
with $\psi (k,t)=1+O(|k|^2)$ as $|k|\to 0$ for any $t\geq 0$.

Therefore, we can consider solutions to \eqref{eq:BoltFourier2} for which the initial datum satisfies 
\be\label{eq:expindata}
 \vphi(k,0)= 1+\sum_{\vert \alpha \vert = 2}c_{\alpha} k^\alpha+ O(\vert k\vert^p) \quad  k\in\R^d, \; p>2.
\ee

In this case we can expect to have solutions to \eqref{eq:BoltFourier2} with the form \eqref{eq:PssGen} with $U=0$. More precisely, self-similar solutions having the form
\be\label{eq:Pss}
\vphi (k,t)=\Psi (ke^{\beta t}),\quad \beta \in \R.
\ee
We explain now that we cannot expect to have $\beta=0$, i.e. stationary solutions for arbitrary matrices $A$. 

The argument above suggests that the solutions with initial data \eqref{eq:expindata} will satisfy
\be\label{eq:expffi_f1}
 \vphi(k,t)= 1+\sum_{\vert \alpha \vert = 2}c_{\alpha}(t) \ k^\alpha+ O(\vert k\vert^p) \quad  t\geq 0,\; k\in\R^d, \; p>2.
\ee
Since $\vert \alpha \vert = 2$ it is more convenient to rewrite \eqref{eq:expffi_f1} as 
\be\label{eq:expffi_f2}
 \vphi(k,t)= 1-\frac 1 2 \sum_{j,\ell=1}^{d} b_{j,\ell}(t) \ k_j k_{\ell}+ O(\vert k\vert^p) \quad  t\geq 0,\; k\in\R^d, \; p>2.
\ee
where $ b_{j,\ell}= b_{\ell,j}$. 
Then, plugging \eqref{eq:expffi_f2} in \eqref{eq:BoltFourier2} (see Appendix A for details) we obtain 
\be \label{eq:SecMomEq}
\partial_t B+ \big(BA +(BA)^T\big)+\frac{q d}{2(d-1)}  \left( B - \frac{\Tr(B)}{d}I\right)=0
\ee 
where $B=(b_{j,\ell})_{j,\ell=1}^d$ and
\be \label{eq:defq}
q=\int_{S^{d-1}} dn   \ g(\hat{k}\cdot n)    \big(  1- (\hat{k}\cdot n)^2  \big) .
\ee

We stress that  \eqref{eq:SecMomEq} is a linear equation for the second order tensor $B$.  
Notice that for an arbitrary matrix $A$ we cannot expect to have convergence of the tensor $B$ towards a constant  tensor $\bar{B}$ as $t\to\infty$, as it would happen if $\vphi(k,t)$ approaches a steady state as $t\to\infty$.  Actually, the theory of linear ODEs implies that $B$ behaves asymptotically as an exponential function $\text{Re}\left( t^a e^{bt}\right)$ with $b\in\C$ and $a\in\N_{\star}$ as $t\to\infty$. We prove later that for $\|A\|$ sufficiently small there exists $\beta\in\R$ such that $b_{j,\ell}(t)\sim N_{j,\ell}e^{2\beta t}$ as $t\to\infty$ (see \cite{JNV}, \cite{TM} for related arguments). Moreover, $\beta$ is small if $\|A\|$ is small. Therefore, we will obtain the asymptotic behaviour
\be \label{eq:expffi_f3}
 \vphi(k,t)= 1-\frac 1 2\sum_{j,\ell=1}^{d} N_{j,\ell}\ \big(e^{\beta t}k_j \big)\big(e^{\beta t}k_{\ell}\big)+ O(\vert k\vert^p) \quad k\in\R^d, \;p>2,
\ee
as $|k|\to 0$. This suggests that the long-time asymptotics of $\vphi(\cdot,t)$ will be given by solutions with the form \eqref{eq:Pss}.
These solutions will play the role of attractors for a large class of solutions of the problem \eqref{eq:phifk0}-\eqref{eq:BoltFourier2} as we will see in Section \ref{sec:8}.

Since it is simpler to work with stationary solutions we will make the following substitution in \eqref{eq:BoltFourier2}:  
\be \label{eq:chvarSSPtoSS}
\vphi (k,t)= \tilde{ \vphi }(\tilde{k},t), \quad \tilde k= k e^{\beta t}.
\ee 
Then, dropping the tildes for convenience, \eqref{eq:BoltFourier2} becomes
\be\label{eq:CauchyAb}
\partial_t\vphi+\big(A_{\beta}k\big) \cdot \partial_k\vphi+\vphi=\Gamma(\vphi), \quad A_{\beta}=A+{\beta}I.
\ee 
The previous considerations suggest that for any small matrix $A$ there exists $\beta\in \R$ such that there exists a stationary 
 solution $\Psi(k)$ to \eqref{eq:CauchyAb}. Moreover, this solution should be an attractor for all the initial data satisfying \eqref{eq:expindata}.  For initial data $\vphi_0(k)$ with $\partial_k \vphi_0(0)\neq 0$ a phase term analogous to the one appearing in \eqref{eq:PssGen} should be included. We will prove all these results in the rest of the paper.

 Given that \eqref{eq:CauchyAb} has the same form of the original equation \eqref{eq:BoltFourier2}, just replacing the matrix  $A$ by $A_{\beta}$, we can write the mild formulation of \eqref{eq:CauchyAb} as we did in the derivation of \eqref{eq:IntBoltFou}. Hence we have the following integral formulation: 
 \be \label{eq:evPhiEbeta}
 \vphi(k,t)=E_{\beta}(t) \vphi_0(k)+\int_0^t d\tau E_{\beta}(\tau) \Gamma \big( \vphi (k,t-\tau) \big),
 \ee
 where 
 \be \label{eq:Ebeta}
 E_{\beta}(t)=\exp \left[ -t\big( 1+ A_{\beta} k\cdot \partial_k\big)\right].
 \ee
Whence, 
  \be \label{eq:EbetaPhi}
 E_{\beta}(t)\vphi_0(k)=e^{-t} \vphi_0 \big(e^{-tA_{\beta} }k\big).
 \ee
In particular, if $\Psi(k)$ is a steady state to  \eqref{eq:evPhiEbeta} we would have 
 \be \label{eq:evPsiEbeta}
 \Psi(k)=E_{\beta}(t) \Psi (k)+\int_0^t d\tau E_{\beta}(\tau) \Gamma \big( \Psi (k) \big). 
 \ee
Since $\vert \Psi (k)\vert \leq 1$ for all $k\in\R^d$, using \eqref{eq:EbetaPhi}  we obtain that $E_{\beta}(t) \Psi (k)$ is exponentially small as $t\to \infty$, and thus taking the limit of \eqref{eq:evPsiEbeta} as $t\to\infty$ we get 
 \be \label{eq:intPsiEbeta}
 \Psi(k)=\int_0^{\infty} d\tau E_{\beta}(\tau) \Gamma[\Psi(k)].
\ee
Reciprocally, every solution to \eqref{eq:intPsiEbeta} is a steady state to \eqref{eq:evPhiEbeta}. Indeed
\bes 
E_{\beta}(t) \Psi(k)= \int_0^{\infty} d\tau E_{\beta}(t+\tau)\Gamma[\Psi(k)]= \int_{t}^{\infty} d\tau E_{\beta}(\tau) \Gamma[\Psi(k)]
\ees
then, writing
\bes 
 \Psi(k)=\int_{t}^{\infty} d\tau E_{\beta}(\tau) \Gamma[\Psi(k)] + \int_0^{t} d\tau E_{\beta}(\tau) \Gamma[\Psi(k)]
\ees 
it follows that $\Psi$ solves \eqref{eq:evPsiEbeta}.

\section{Existence and uniqueness of stationary solutions to $(\ref{eq:CauchyAb})$ } \label{sec:7}
We first introduce an eigenvalue  problem that will play a crucial role in what follows:
\be \label{eq:EVpbMom}
2 \beta B+ \big(BA +(BA)^T\big)+\frac{q d}{2(d-1)}  \left( B - \frac{\Tr(B)}{d}I \right)=0
\ee 
with $q$ given as in \eqref{eq:defq} and $B=(b_{j,\ell})_{j,\ell=1}^d \in \text{Symm}_{d\times d}(\C)$ which is the space of $d\times d$ symmetric matrices with complex coefficients. The rationale to study this eigenvalue problem is the analysis of the equations for the second order moments in \eqref{eq:SecMomEq}. More precisely, we obtain \eqref{eq:EVpbMom} looking for solutions of \eqref{eq:SecMomEq}  with the form $e^{2\beta t} B$ where $B$ is a fixed symmetric matrix independent on time. 

We also introduce the sets of functions
\be\label{def:FpB}
\mathcal{F}_p(B):=\left\{\Psi \in \Phi\;\; \text{s.t.}\; \; \sup_{k\in\R^d} \frac{\Big\vert \Psi(k)-\Big(1-\frac 1 2 \sum_{j,\ell=1}^d b_{j,\ell}k_j k_{\ell} \Big)\Big \vert}{\vert k\vert^p} < \infty\right\}
\ee
where $B \in \text{Symm}_{d\times d}(\C)$ and $p>2$.  Notice that it is not guaranteed a priori if these sets are nonempty because the matrix $B$ can be incompatible with $\Psi$ being a characteristic function. Assuming that $\mathcal{F}_p(B)\neq \emptyset$ we can endow it with a structure of metric space using the distance
\be \label{def:metricFpB}
d(\Psi_1,\Psi_2)=\sup_{k\in\R^d} \frac{\Big\vert \Psi_1(k)-\Psi_2(k)\Big \vert}{\vert k\vert^p}.
\ee
We further notice that for any $q>p>2$ we have $\mathcal{F}_q(B)\subset \mathcal{F}_p(B)$ since 
$\vert \Psi(k) \vert \leq \Psi(0)=1$ due to the fact that $\Psi\in\Phi$. 

We now prove the following result. 
\begin{theorem} \label{th:main1}
Let be $p>2$. There exists $\ep_0>0$ depending on the collision kernel $g$ in \eqref{eq:CollBolt} and $p$, such that if $\|A\| \leq \ep_0$  the following results hold. 
\begin{itemize}
\item[(i)] The eigenvalue $\beta$ solution of \eqref{eq:EVpbMom} with the largest real part is real and simple. 
The corresponding eigenvector which will be denoted as $N=(N_{j,\ell})_{j,\ell=1}^d$ has real coefficients, i.e. $N\in \text{Symm}_{d\times d}(\R)$, $N$ is a positive definite matrix and it can be normalized as $\frac 1 d  \Tr\big(N^T N \big)=1$ .
\item[(ii)]  Let be $\beta \in \R$ and $N \in \text{Symm}_{d\times d}(\R)$ as in item $(i)$. Then, the set $\mathcal{F}_p(B)$ with $B=N$ and $2<p\leq 4$ in \eqref{def:FpB} is nonempty. 
Moreover, there exists a unique $\Psi(\cdot)\in \mathcal{F}_p(N)$ that solves \eqref{eq:intPsiEbeta}, namely
 \be \label{eq:intPsiEbeta2}
 \Psi(k)=\int_0^{\infty} d\tau E_{\beta}(\tau) \Gamma[\Psi(k)]
\ee
with $E_\beta$ as in \eqref{eq:Ebeta}.
\end{itemize}
\end{theorem}

\begin{remark}
We observe that if $\Psi(\cdot)\in \mathcal{F}_p(N)$ is a solution to  \eqref{eq:intPsiEbeta2}, then $\Psi(\lambda \cdot) \in \mathcal{F}_p(\lambda^2 N)$ and it solves also \eqref{eq:intPsiEbeta2}.
\end{remark}

\begin{remark}
Notice that the spaces $\mathcal{F}_p(N)$ with $p\to 2^{+}$ are larger than the spaces $\mathcal{F}_p(N)$ with larger values of $p$. On the other hand, in order to obtain existence and uniqueness of solutions in the space $\mathcal{F}_p(N)$ with $p\to 2^{+}$ we have to assume more stringent smallness conditions for $\|A\|$. We further observe that the analysis of the solutions to \eqref{eq:intPsiEbeta} when $p\to 2^{+}$ and $\|A\|\to 0$ is a non trivial problem that however will not be considered in this paper. 
\end{remark}

In order to prove this Theorem we would need some preliminary results. 
\begin{lemma} \label{lem:EPbeta0}
The following results hold.
\begin{itemize}
\item[(i)]There exists $\ep_0>0$ depending on the dimension $d$ such that if $\|A\| \leq \ep_0$ the eigenvalue problem \eqref{eq:EVpbMom} has a unique  eigenvalue $\beta_0\in\R$ which is simple and has the property that $\text{Re}(\beta)<\beta_0$ where $\beta$ is any other eigenvalue of \eqref{eq:EVpbMom}. 
Furthermore we will denote as $N$ the corresponding eigenvector to the eigenvalue $\beta_0$ as $N$. The matrix $N$ has real coefficients and it will be normalised in such a way that  that $\frac 1 d  \Tr\big(N^TN\big)=1$. 

\item[(ii)] Under the same assumptions of item $(i)$, there exists a constant $C_0\in\R$ depending on $d$ such that  $\vert \beta_0 \vert \leq C_0 \|A\|$ and  $\| N-I \|\leq C_0 \|A\|$.

\end{itemize}
\end{lemma}
\begin{proof}
The proof follows the same strategy of the proof of  Lemma 4.16 in \cite{JNV}. 

If $A=0$ the eigenvalue problem \eqref{eq:EVpbMom} reduces to 
 \bes
2 \beta B+\frac{q d}{2(d-1)}  \left( B - \frac{\Tr(B)}{d} I \right)=0,\quad B\in\text{Symm}_{d\times d}(\C).
\ees
There are two eigenvalues for this problem, namely $\beta=0$ and $\beta= \frac{\Tr(B)}{d} I$. The eigenvalue $\beta=0$ is simple with corresponding eigenspace $\{cI\;:\; c\in\C\}$. Then the lemma follows using standard continuity and differentiability results for the spectrum of finite dimensional eigenvalue problems (see for instance \cite{K76}).
\end{proof}

The following technical result will be required later. 
\begin{lemma}\label{lem:techlem}
Let $\beta \in \R$ and $A\in M_{d\times d}(R)$ satisfy the inequality $\displaystyle \vert| A \vert| < \frac {1}{2}+\beta$. 
Then, the following equation  
\be \label{eq:matrixeq}
(1+2\beta) M+MA+\big(MA\big)^T=0
\ee
for the matrix $M\in M_{d\times d}(R)$ has only the trivial solution $M=0$.   
\end{lemma}

\begin{proof}
The first inequality implies that $\beta>-\frac 1 2$. Then, from \eqref{eq:matrixeq}, we obtain 
\bes
(1+2\beta) \|M\|\leq \vert| MA+\big(MA\big)^T\vert|\leq 2 \vert| M\vert| \ \vert| A\vert|.
\ees
If $\vert| M\vert|>0$ we obtain the inequality $\frac 1 2 + \beta \leq \vert| A \vert|$ that contradicts the assumption of the lemma. Hence, this completes the proof.
\end{proof}

The following lemma provides an approximate solution $\Psi_0(k)$ of the equation \eqref{eq:intPsiEbeta2}. It will be used to show that the set $\mathcal{F}_p(B)$ defined in \eqref{def:FpB} is non empty for $p\leq 4$.

\begin{lemma}\label{lem:StaSol}
Suppose that $\|A\|$ is small enough to guarantee that item $(i)$ of Lemma \ref{lem:EPbeta0} holds. 
There exists $\ep_1>0$ depending on the dimension $d$ and on $q$ (cf.~\eqref{eq:defq}) such that if $\|A\|\leq \ep_1$  
then the function 
\be \label{eq:Psi0efN}
\Psi_0(k)=\exp\left[-\frac 1 2 N : k\otimes k\right]\quad \text{where}\;\;N : k\otimes k=N_{j,\ell} k_j k_\ell 
\ee
is contained in $\Phi$. Moreover, there exists a constant $C_1\in
\R$ depending only on $d$ such that
\be \label{eq:estPsi0StSo}
\left\vert \Psi_0(k)-\int_0^{\infty} d\tau E_{\beta}(\tau) \Gamma \big[\Psi_0(k)\big] \right\vert \leq C_1 \min\{1,\vert k \vert^4\},\quad k\in\R^d
\ee
where $E_{\beta}(t)$ is as in \eqref{eq:Ebeta}.
\end{lemma}

\begin{proof}
We first observe that, item $(ii)$ in Lemma \ref{lem:EPbeta0} implies that,  if $\ep_1$ is sufficiently small, the matrix $N$ yields a positive definite quadratic form. Then, $\Psi_0(k)$ is a gaussian function and its inverse Fourier Transform is a gaussian as well. Moreover, it gives a probability density since $\Psi_0(0)=1$. Then, $\Psi_0\in\Phi$.

Arguing as in the derivation of \eqref{eq:rhsMom} in Appendix \ref{app:A} just replacing $\vphi(k)$ by $\Psi_0(k)$ and the matrix $B$ by $N$ we obtain
\be\label{eq:rhsMomPsi}
\Gamma[\Psi_0](k)= \Psi_0(k)+ \frac{q d}{4(d-1)}  N : \big( k \otimes k - \frac 1 d \vert k \vert^2 I \big)+O(\vert k\vert^4)
\ee
where the constant estimating $O(\vert k\vert^4)$ depends only on $d$ (because $\|A\|$ and $\beta$ are assumed to be small).

We now claim that for any function $F\in C^{\infty}(\R^d)$ the function $G$ defined by means of
\be\label{eq:defGk}
G(k)=\int_0^{\infty}d\tau E_{\beta}(\tau)F(k)
\ee 
is a function in $C^{\infty}(\R^d)$  and it satisfies 
\be \label{eq:evGk}
\big(A_{\beta}k\big) \cdot \partial_k G(k)+G(k)=F(k).
\ee
We notice that using  \eqref{eq:defGk} as well as \eqref{eq:EbetaPhi} we obtain 
\bes
G\big(e^{-A_{\beta} h }k\big)=\int_0^{\infty} d\tau e^{-\tau} F\big(e^{-A_{\beta} (\tau+h) }k\big) ,\quad \forall \; k\in\R^d,\; h\in\R \setminus \{0\}.
\ees
Then, using the change of variables $\tau+h\to \tau$ we get
\begin{align*}
G\big(e^{-A_{\beta} h }k\big)&=e^h \int_0^{\infty} d\tau e^{-\tau} F\big(e^{-A_{\beta} \tau }k\big) - e^h \int_0^{h} d\tau e^{-\tau} F\big(e^{-A_{\beta} \tau }k\big) \\&
= e^hG\big(k\big)-e^h \int_0^{h} d\tau e^{-\tau} F\big(e^{-A_{\beta} \tau }k\big) .
\end{align*}
Then, performing a Taylor expansion for small $h$, and using in particular that $e^{-A_{\beta} h }= I-A_{\beta} h + O(h^2)$ we obtain  
\bes
G(k) - h (A_{\beta} k)\cdot \partial_k G(k)=G(k)+ hG(k)-hF(k)+O(h^2).
\ees
Cancelling the term $G(k)$ in both sides, dividing by $h$ and taking the limit as $h\to 0$ we obtain \eqref{eq:evGk}.

Using \eqref{eq:rhsMomPsi} we obtain
\begin{align*}
J(k) &:= \Psi_0(k)-\int_0^{\infty} d\tau E_{\beta}(\tau) \Gamma \big[\Psi_0(k)\big] \\&
= \Psi_0(k)-\int_0^{\infty} d\tau E_{\beta}(\tau) \Psi_0(k) + G_1(k) +O(\vert k\vert^4)
\end{align*}
where 
$$G_1(k)=\int_0^{\infty} d\tau E_{\beta}(\tau)F_1(k)$$ 
with $F(k)$ given by the quadratic form, hence the $C^{\infty}$ function, 
$$F_1(k)=- \frac{q d}{4(d-1)} N : \big( k \otimes k - \frac 1 d \vert k \vert^2 I \big).$$ 
Using the Taylor expansion at $k=0$ of the function $\Psi_0(k)$ given as in \eqref{eq:Psi0efN} we obtain
\begin{align*}
J(k) = 1- \frac{1}{2} N : \big( k \otimes k \big)-1+\int_0^{\infty} d\tau E_{\beta}(\tau) F_2(k) + G_1(k) +O(\vert k\vert^4)
\end{align*}
where 
$$F_2(k)= \frac{1}{2} N : \big( k \otimes k \big).$$ 
Hence
\be\label{eq:defJk}
J(k) = - \frac{1}{2} N : \big( k \otimes k \big)+ G(k) +O(\vert k\vert^4)
\ee
where $G(k)$ is as in \eqref{eq:defGk} with $F(k)=F_1(k)+F_2(k)$. Moreover, since $F$ is a quadratic form it follows that also $G$
 is a quadratic form as it may be seen using \eqref{eq:defGk} and \eqref{eq:EbetaPhi}.

With a slight abuse of notation we will write $G(k)=G: (k\otimes k)$ with $G\in M_{d\times d}(\R)$. Then, using the form of $F(k)$ above as well as the definition of $A_\beta$ given in \eqref{eq:CauchyAb}, equation \eqref{eq:evGk} becomes
\be \label{eq:evGk2}
(1+2 \beta) G+ \big(GA +(GA)^T\big)=-\frac{q d}{4(d-1)}  \left( N  - \frac{\Tr(N )}{d}I \right) + \frac 1 2 N. 
\ee
We observe that \eqref{eq:evGk2} is a linear, non homogeneous equation for $G$ and therefore, its solution can be written as a sum of particular solution plus the solution of the homogeneous equation.
 A particular solution of \eqref{eq:evGk2} is $G_p=\frac 1 2 N$ as it might be seen using that $N$ solves \eqref{eq:EVpbMom}. 
 Then the solution to \eqref{eq:evGk2} is given by $G=\frac 1 2 N+M$ where $M$ is a solution of the homogeneous problem 
 \bes 
 (1+2 \beta) M+ \big(MA +(MA)^T\big)=0.
 \ees
 Lemma \ref{lem:techlem} shows that $M=0$  if $\ep_1$ is sufficiently small.  Hence $G(k)= \frac 1 2 N : k\otimes k$.  Therefore, \eqref{eq:defJk} implies $J(k)=O(\vert k \vert ^4)$ for any $k\in\R^d$. Using then that $\| \Psi(k)\|_{\infty} \leq 1 $ and $\|\Gamma(\Psi)\|_{\infty}\leq 1$ we obtain  \eqref{eq:estPsi0StSo} and the Lemma \ref{lem:StaSol} follows
\end{proof}

We can now prove Theorem \ref{th:main1}.

\begin{proofof}[Proof of Theorem \ref{th:main1}]
We first notice that item $(i)$ follows from Lemma \ref{lem:EPbeta0}. 

Concerning item $(ii)$ we observe that the fact that $\mathcal{F}_p(N)$ with $2<p\leq 4$ is nonempty as a consequence of Lemma \ref{lem:StaSol} since $\Psi_0(k)$ in \eqref{eq:Psi0efN} belongs to $\mathcal{F}_p(N)$. Moreover, $\mathcal{F}_p(N)$  endowed with the metric \eqref{def:metricFpB} is a complete metric space. This follows from the fact that the convergence in the topology induced by the metric $d$ in \eqref{def:metricFpB} implies uniform convergence in compact sets of $\R^d$. The space $\Phi$ of characteristic functions is closed under the uniform convergence in compact sets (cf.~\cite{Fe2}). It only remains to show that there exists a function $\Psi\in \mathcal{F}_p(N)$ satisfying \eqref{eq:intPsiEbeta2}. To this end we use a fixed point argument. 
We define an operator $\mathcal{S}:\mathcal{F}_p(N)\rightarrow \mathcal{F}_p(N)$ such that 
\be\label{eq:defS}
\big(\mathcal{S}\Psi\big)(k)=\int_0^{\infty} d\tau E_{\beta}(\tau) \Gamma[\Psi(k)].
\ee 
In order to prove Theorem \ref{th:main1} we now prove the following estimate for the operator $\mathcal{S}$. 
Let be $\Psi_1,\ \Psi_2\in \mathcal{F}_p(N)$. We have 
\be \label{eq:contrS}
\Big\vert \mathcal{S}\Psi_1-\mathcal{S}\Psi_2\Big\vert (k)\leq \theta d(\Psi_1,\Psi_2) \vert k\vert ^p
\ee
where $d$ is the distance defined as in \eqref{def:metricFpB} and $0<\theta<1$ if $\ep_0$ is sufficiently small.

To show \eqref{eq:contrS} we observe that using \eqref{eq:defS} as well as the $\mathcal{L}$-Lipschitzianity property (cf.~Lemma \ref{lm:bdI}) we obtain
\begin{align*}
\Big\vert \mathcal{S}\Psi_1-\mathcal{S}\Psi_2\Big\vert (k)&\leq \int_0^{\infty} d\tau  E_{\beta}(\tau) \mathcal{L}(\vert \Psi_1-\Psi_2 \vert )(k) \\&
\leq d(\Psi_1,\Psi_2) \int_0^{\infty} d\tau e^{-\tau} e^{-\tau (Ak)\cdot \partial_k}  \mathcal{L}(\vert k\vert^p) \\& = (1-\lambda(p)) d(\Psi_1,\Psi_2) \int_0^{\infty} d\tau e^{-\tau}  \vert e^{-\tau (Ak)\cdot \partial_k}  k\vert^p\\&
=  \theta d(\Psi_1,\Psi_2)\vert k\vert^p
\end{align*} 
where in the first equality we used \eqref{eq:Lkp} and we set $\theta:=\frac{(1-\lambda(p)) }{1-p\|A\|}$.  Note that $\theta<1$ if $\|A\| < \frac{1}{p}$. We remark that we are using also that $\lambda(p)>0$ if $p>2$ (cf.~Theorem \ref{thm:3}). Hence the claim \eqref{eq:contrS} follows.

We need to check that $\mathcal{S}$ maps $\mathcal{F}_p(N)$ into itself. The fact that $\big(\mathcal{S}\Psi\big)\in\Phi$ for any $\Psi \in\mathcal{F}_p(N)$ follows from Lemma \ref{eq:I+map}. We now prove that $\mathcal{S}\Psi \in \mathcal{F}_p(N)$. 

Due to \eqref{def:FpB} and \eqref{eq:Psi0efN} to prove that a function $\Psi\in\Phi$ belongs to $\mathcal{F}_p(N)$ with $p\in (2,4]$ it is equivalent to show that 
$$\sup_{k\in\R^d}\frac{\Big\vert \Psi(k)-\Psi_0(k)\Big \vert}{\vert k\vert^p}<\infty. $$

Therefore, in order to show that given $\Psi\in \mathcal{F}_p(N)$ with $p\in (2,4]$ we have $\mathcal{S}\Psi \in \mathcal{F}_p(N)$ with $p\in (2,4]$ we can use the following chain of inequalities
\begin{align*}
\sup_{k\in\R^d}\frac{\Big\vert \mathcal{S}\Psi(k)-\Psi_0(k)\Big \vert}{\vert k\vert^p} &\leq 
\sup_{k\in\R^d} \frac{\Big\vert \mathcal{S}\Psi(k)-\mathcal{S}\Psi_0(k)\Big \vert}{\vert k\vert^p} + C  \leq \theta d(\Psi,\Psi_0) +C 
\end{align*} 
for some finite constant $C$. Note that we used in the first inequality  \eqref{eq:estPsi0StSo} and in the second the claim \eqref{eq:contrS} and \eqref{eq:defS}.
Thus  $\mathcal{S}\Big(\mathcal{F}_p(N)\Big)\subset \mathcal{F}_p(N)$ for any $2<p\leq 4$.

It only remains to prove that the operator $\mathcal{S}$ is contractive in $\mathcal{F}_p(N)$. Using \eqref{eq:contrS} we obtain
\bes
d(\mathcal{S}\Psi_1, \mathcal{S}\Psi_2) \leq \theta d(\Psi_1,\Psi_2).
\ees
Therefore, Banach fixed point theorem ensures the existence of a unique fixed point in $\mathcal{F}_p(N)$ for the operator $\mathcal{S}$.
\end{proofof}


\section{Stability of the steady states of $(\ref{eq:evPhiEbeta})$ } \label{sec:8}

In this section we prove the following result for time-dependent solutions of the problem \eqref{eq:CauchyAb} (cf.~\eqref{eq:evPhiEbeta}) having finite energy. 

\begin{theorem} \label{th:main2}
Let $p>2$ and $\|A\|\leq \ep_0$ with $\ep_0$ as in  Theorem \ref{th:main1}.   Let $\beta \in \R$ be as in Theorem \ref{th:main1} and let $\Psi(\cdot)$ the solution of  \eqref{eq:intPsiEbeta} described in item $(ii)$ of Theorem \ref{th:main1}. Suppose that $\vphi_0\in\Phi$ satisfies
\be\label{eq:cdtffi0}
\vert  \vphi_0(k) - \big(1+\sum_{\vert \alpha \vert = 2}c_{\alpha} k^\alpha \big)\vert \leq C_0 \vert k\vert^p, \quad k \in \R^d
\ee
for some $C_0>0$.  Let $\vphi(\cdot,t)$ be the solution to \eqref{eq:evPhiEbeta}.
Then, there exists $\lambda \geq 0$, $C>0$ depending on $\vphi_0$ and $\theta>0$ depending on $q$ as in \eqref{eq:defq}, $p$ and $d$ such that 
\be \label{eq:stabin}
\vert  \vphi(k,t)-\Psi (\lambda k)\vert \leq C e^{-\theta t} (|k|^2+|k|^p), \quad k \in \R^d. 
\ee 
\end{theorem}

\begin{remark}
Notice that the argument at the end of Section \ref{sec:6} shows that $\Psi(\cdot)$ is a stationary solution to  \eqref{eq:CauchyAb}, or equivalently to \eqref{eq:evPhiEbeta}. 
\end{remark}

In order to prove Theorem \ref{th:main2} we will need some auxiliary results. The following lemma provides the long-time asymptotics of the Hessian matrix of $\vphi(\cdot,t)$ at $k=0$. 
\begin{lemma}\label{lem:asym}
Let $\beta \in \R$ be as in the statement of Theorem \ref{th:main1}.  Suppose that $\vphi \in C([0,\infty);\Phi)$ is a solution to \eqref{eq:evPhiEbeta} with initial datum $\vphi_0\in\Phi$ satisfying \eqref{eq:cdtffi0}.  
Let $B(t)\in \text{Symm}_{d\times d}(\R)$ such that 
\bes
\vert  \vphi(k,t) - \big(1- \frac 1 2 B(t) k \otimes k \big)\vert \leq C_0(t) \vert k\vert^p, \quad k \in \R^d, \;\;p>2
\ees
(cf.~Proposition \ref{prop:regularity}). Then, $B(t)$ solves the linear equation 
\be \label{eq:SecMomEq2}
\partial_t B+ \big(BA_{\beta} +(BA_{\beta})^T\big)+\frac{q d}{2(d-1)}  \left( B - \frac{\Tr(B)}{d}I\right)=0
\ee 
where $q$ is given by \eqref{eq:defq}. Thus, there exists a real constant $\lambda\geq 0$ such that the asymptotic behaviour of $B(t)$ is given by 
\be \label{eq:asymB}
B(t) \rightarrow \lambda^2 N \;\; \text{as} \; t\to\infty
\ee
where $N$ is as in item $(i)$ of Theorem \ref{th:main1}. Moreover, the convergence is exponential, i.e.
\be\label{eq:expconv}
\|B(t) - \lambda^2 N\|\leq C e^{-\nu q t}
\ee
where $q$ is as in \eqref{eq:defq}, the constant $C$ depends only on the moments $c_{\alpha}$ on the left hand side of  \eqref{eq:cdtffi0} and $\nu$ depends only on the dimension $d$ but it is independent on $A$.
\end{lemma}
\begin{remark}\label{rk:dirac}
Notice that $\kappa_0=0$ (or, equivalently $\lambda=0$ in \eqref{eq:stabin}) if and only if $B(t)=0$ for any $t\geq 0$. Then $\vphi(k,t) =1$ for any $t\geq 0$ since $\vphi\in\Phi$ (the corresponding probability distribution  $f$ is a Dirac mass at the origin).
\end{remark}

\begin{proof}
The fact that $B(t)$ solves \eqref{eq:SecMomEq2} follows from Proposition  \ref{prop:regularity} and from arguments similar to the ones yielding \eqref{eq:SecMomEq}. 

The result concerning the long-time asymptotics of $B(t)$ is a consequence of the classical theory of ODEs. Notice that the choice of $\beta$ in Theorem \ref{th:main1} implies that $B=N$ is a steady state of \eqref{eq:SecMomEq2} and that all the other eigenvalues associated to the solution of the linear equation \eqref{eq:SecMomEq2} have negative real part. Therefore $B(t)$ approaches $ \kappa_0 N$ for some $ \kappa_0 \in\R$ exponentially (cf.~\eqref{eq:expconv}). The dependence of $\nu$ only on the dimension $d$ follows from the fact that for $\|A\|$ small the rate of convergence of $B(t)$ to the hyperplane of steady states to \eqref{eq:SecMomEq2} is close to the one for the equation
\bes
\partial_t B+\frac{q d}{2(d-1)}  \left( B - \frac{\Tr(B)}{d}I\right)=0.
\ees
Absorbing $q$ in the time variable we then obtain \eqref{eq:expconv}. 
We also observe that $\kappa_0 \geq 0$ because $B(t)$, as well as $N$ are positive definite matrices. In the case of $B(t)$ this is due to the fact that $\vphi(\cdot,t)\in\Phi$ and then  $B(t)$ is the covariance matrix associated to the corresponding probability distribution $f$. We can then write $\kappa_0=\lambda^2$ for definiteness. 
\end{proof}

The following result yields a  Gaussian approximation for the solution of \eqref{eq:CauchyAb}.
\begin{lemma}\label{lem:approx}
Let $\vphi\in C([0,\infty);\Phi)$ be a solution to \eqref{eq:CauchyAb}.
Let $B(t)$ as in Lemma \ref{lem:asym}. 
We define the function $\tilde{\vphi}\in C([0,\infty);\Phi)$ as 
\be \label{eq:exptildeffi}
\tilde{\vphi}(k,t)=e^{-\frac 1 2 B(t):k\otimes k}, \quad k\in\R^d.
\ee
Then $\tilde{\vphi}$ satisfies the following equation
\be
\left\vert \partial_t \tilde{\vphi}+(A_{\beta}k)\cdot \partial_k \tilde{\vphi}-\Gamma [\tilde{\vphi}]+\tilde{\vphi}\right\vert \leq C_1 \vert k\vert^4, \quad k\in\R^d
\ee 
where $C_1\in \R_{+} $ depends on the coefficients $c_{\alpha}$ in \eqref{eq:cdtffi0} but it is independent on $t$. 
\end{lemma}

Notice that $\vphi$ and $\tilde{\vphi}$ have the same Hessian matrix at $k=0$ for any time $t$. 

\begin{proof}
The result follows using a Taylor series for the exponential function \eqref{eq:exptildeffi} as well as the equation \eqref{eq:SecMomEq2} in order to cancel the quadratic terms. Notice that the constant in front of the term $\vert k\vert^4,$ depends only on $B(0)$ and $d$. See Lemma \ref{lem:StaSol} for a related argument. 
\end{proof}

\begin{lemma} \label{lem:ffiffitilde}
Let $\vphi$ and $\tilde{\vphi}$ be given as in Lemma \ref{lem:approx}. The following estimate holds:
\be
\vert  \vphi(k,t)-\tilde{\vphi}(k,t)\vert \leq  C_2 \vert k\vert^p, \quad t \geq 0, \quad  2<p\leq 4,
\ee
where $C_2$ depends on $B(0)$, on $C_0$ (cf.\eqref{eq:cdtffi0}), $d$ and $p$, but it is independent on $t$. 
\end{lemma}

\begin{proofof}[Proof of Lemma \ref{lem:ffiffitilde}]
We consider the equation satisfied by the difference $\vphi(k)-\tilde{\vphi}(k)$, namely
\bes 
 \partial_t \big(\vphi-\tilde{\vphi}\big)+(A_{\beta}k)\cdot \partial_k \big(\vphi-\tilde{\vphi}\big)=\Gamma [\vphi]-\Gamma [\tilde{\vphi}]-\big(\vphi-\tilde{\vphi}\big)+O(\vert k\vert^4).
\ees 
Using  the definition of $B(t)$ in Lemma \ref{lem:asym} and \eqref{eq:exptildeffi} we have
\bes
\vert \vphi (k, 0)-\tilde{\vphi}(k, 0) \vert \leq C \vert k \vert ^p, \quad k\in \R^d 
\ees 
Then, a standard comparison argument (cf. Theorem \ref{thm:2} for a similar result) gives
\bes
\vert \vphi  (k,t)-\tilde{\vphi} (k,t) \vert \leq C \exp  \left[\left( \mathcal{L}-\hat{D}-I\right)t \right]  \vert k \vert ^p+  C_1 \int_0^t ds  \exp  \left[\left( \mathcal{L}-\hat{D}-I\right)(t-s) \right]\vert k \vert^4.
\ees 
Using Proposition \ref{prop:supersol} and Lemma \ref{lem:bdup} we obtain 
\bes 
\vert \vphi  (k,t)-\tilde{\vphi} (k,t) \vert \leq C u_p (k,t)+ C_1 \int_0^t ds \ u_p (k,t-s),
\ees
 where $u_p$ is defined as in \eqref{eq:defup}. Therefore $\vert \vphi  (k,t)-\tilde{\vphi} (k,t) \vert \leq  C_2 \vert k\vert^4$ and the result follows.
\end{proofof}

\begin{proofof}[Proof of Theorem \ref{th:main2}]
Using that $\Psi\in \mathcal{F}_p(N)$  with $\mathcal{F}_p(N)$ as in \eqref{def:FpB},  we obtain  
\be \label{estPTh}
\vert \vphi(k,T)- \Psi(\lambda k)\vert \leq \vert \vphi(k,T)- \tilde{\vphi}(k,T)\vert+ \vert  \tilde{\vphi}(k,T)- \Psi(\lambda k)\vert \leq  C|k|^p + C e^{-\nu q T} |k|^2
\ee 
with $2<p\leq 4$ and where we have used  Lemma \ref{lem:ffiffitilde} to estimate the first term on the right hand side of the first inequality and \eqref{eq:expconv} in Lemma \ref{lem:asym} to estimate the second term on the right hand side. 

We now apply again Theorem \ref{thm:2} starting at time $t=T$. Using also Proposition \ref{prop:supersol} and Lemma \ref{lem:bdup}, the following estimate holds:
\bes 
\vert \vphi(k,T+L)- \Psi(\lambda k)\vert \leq Ce^{-(\lambda(p)-p\|A\|) L}|k|^p+ C e^{-\nu q T} e^{ 2 \|A\| L}|k|^2 ,\quad L\geq 0,
\ees 
Notice that we have used $p=2$ to estimate the contribution due to the second term on the right hand side of \eqref{estPTh}. 

Assuming that $\|A\| \leq \min \left\{ \frac {\lambda(p)}{2p},\frac{\nu q}{4} \right\}$ we can take $L=T$ so that  
 \bes 
 \vert \vphi(k,2T)- \Psi(\lambda k)\vert \leq C e^{-2\theta T} (|k|^2+|k|^p), \quad k\in\R^d
 \ees 
where $\theta=\min \left\{ \frac {\lambda(p)}{4},\frac{\nu q}{4} \right\}$.
Hence the proof is completed.   
\end{proofof}

\section{Finiteness of the moments of the stationary solutions to arbitrary order}  \label{sec:9}

We prove in this section that the steady state $\Psi(k)$ defined by means of \eqref{eq:intPsiEbeta} has moments of arbitrary order if we assume that $\|A\|$ is sufficiently small.
 
We consider the stationary problem associated to \eqref{eq:CauchyAb}, namely 
\be\label{eq:s9CauchySt}
\big(A_{\beta}k\big) \cdot \partial_k\vphi+\vphi=\Gamma(\vphi).  
\ee 
We first compute the expected form for the moments $c_{\alpha}$ with $|\alpha|\leq M$ for any given $M$.  To this end we write 
\be\label{eq:s9ffi}
\vphi(k)=1+\sum_{2\leq \ell \leq M}Q_{\ell}(k)+O(|k|^p),\quad M<p \leq M+1,\quad Q_{\ell}\in H_{\ell}(\C)
\ee
where $H_{\ell}(\C)$ is the space of homogeneous polynomials in the variables $k_1,\dots, k_d$ with complex coefficients, namely
\be\label{eq:s9Hl}
H_{\ell}(\C)=\{h\ :\ h(k)=\sum_{|\alpha|=\ell} c_{\alpha}(k)^\alpha,\; c_{\alpha}\in\C\}\quad \ell=0,1,2,\dots
\ee
Notice that these spaces are not contained in $\Phi$ if $\ell \neq 1$.  
We remark that we can assume without loss of generality that the sum in \eqref{eq:s9ffi} starts at $|\alpha|=2$ since we can remove the first order moments multiplying  by a phase factor as in \eqref{eq:PssGen}.

Plugging formally \eqref{eq:s9ffi} into \eqref{eq:s9CauchySt} we get the following equation for the moments of order $M$
\be\label{eq:s9momeq}
Q_{\ell}(k)-\mathcal{L}Q_{\ell}(k)+ (A_{\beta}k)\cdot\partial_{k}Q_{\ell}(k)= \mathbb{K}_{\ell},\quad Q_{\ell}\in H_{\ell}(\C)
\ee
where $\mathcal{L}$ is as in \eqref{eq:linBolzF} and 
\begin{align}
& \mathbb{K}_{\ell}=\mathbb{K}_{\ell}\left[\{Q_j\}_{j=2}^{\ell-1}\right]
=\sum_{m+j=\ell,\; 2\leq m,j\leq M }\int_{S^{d-1}} dn\ g(\hat{k}\cdot n)Q_{j}(k_{+})Q_{m}(k_{-}),\quad \ell\geq 3,\nonumber \\& \mathbb{K}_2=0. \label{eq:s9opK}
\end{align}
We would denote as $\mathcal{L}_{\ell}$ the operator $\mathcal{L}$ restricted to the space $H_{\ell}(\C)$. Notice that $\mathcal{L}\big(H_{\ell}(\C)\big)\subset H_{\ell}(\C)$. Hence the operators $\mathcal{L}_{\ell}$ are just finite dimensional operators $\mathcal{L}_{\ell}: H_{\ell}(\C)\rightarrow H_{\ell}(\C)$. Moreover, we will use  the following norm in the spaces $H_{\ell}(\C)$:
\be
\label{eq:s9norm}
\|\vphi\|_{\ell}:=\sup_{k\in\R^d}\frac{|\vphi(k)|}{|k|^{\ell}}=\sup_{k\in\R^d,\ |k|=1} |\vphi(k)|.
\ee
\medskip
We first prove an auxiliary property of the space $H_{\ell}(\C)$ that will be useful in the following. 
\begin{lemma}\label{lem:hompol}
 Suppose that $P\in H_{\ell}(\C)$ for any $\ell\geq 0,$ has the form $P(k)=\sum_{|\alpha|=\ell} c_{\alpha}(k)^\alpha,$ with $c_{\alpha}\in\C$. Then there exists a constant $C_{\ell}>0$ depending only on the number $\ell$ such that
\be\label{eq:bdpol}
\sup_{\alpha} \vert c_{\alpha}\vert  \leq C_{\ell} \| P\|_{\ell}. 
 \ee
\end{lemma}
\begin{proof}
We observe that the left hand side of \eqref{eq:bdpol} is a norm in the finite dimensional space $H_{\ell}(\C)$.  
Then, since $\| \cdot \|_{\ell}$ is also a norm in $H_{\ell}(\C)$ the result follows  just using the equivalence of all the norms in any finite dimensional normed space. 
\end{proof}

We will show in the following that if $\|A\|$ is sufficiently small, equation \eqref{eq:s9momeq} can be solved inductively. 

Suppose that $Q_2(k)=-\frac 1 2 N:k\otimes k$ where $N$ is as in item $(i)$ of Theorem \ref{th:main1} and that there is a family of homogeneous polynomials $Q_{\ell} \in H_{\ell}(\C)$ for  $\ell= 3,\dots,M$ satisfying \eqref{eq:s9momeq}, \eqref{eq:s9opK}. Let be $p\in (M,M+1]$. 
We then define the following generalization of the space $\mathcal{F}_p(N)$ given in \eqref{def:FpB}:\be\label{def:FMpN}
\mathcal{F}_{p,M}(N):=\left\{\Psi \in \Phi\;\; \text{s.t.}\; \; \sup_{k\in\R^d} \frac{\Big\vert \Psi(k)-\Big(1+ \sum_{\ell=2}^{ M} Q_{\ell}(k) \Big)\Big \vert}{\vert k\vert^p} < \infty\right\}.
\ee
Notice that it is not ``a priori" clear if the space $\mathcal{F}_{p,M}(N)$ is nonempty even if the polynomials $Q_{\ell}$ are given since it is not obvious the existence of a characteristic function $\Psi$ such that the inequality in the definition \eqref{def:FMpN} holds. 
In the rest of this section we will prove the following result.  
\begin{theorem} \label{th:higmom}
Let be $M\in\N$, $M\geq 3$, $p\in(M,M+1]$ and $\|A\| \leq \ep_0$ with $\ep_0>0$  as in Theorem \ref{th:main1}.   
Let then be $\beta \in\R$ and $N \in Symm_{d\times d}(\R)$ as in item $(i)$ of  Theorem \ref{th:main1}.  Then there exists an $\ep_M\in (0,\ep_0)$ such that if $\|A\| \leq \ep_M$  the set $\mathcal{F}_{p,M}(N)$ defined in \eqref{def:FMpN} is nonempty. 
Moreover, there exists a unique $\Psi(\cdot)\in \mathcal{F}_{p,M}(N)$ that solves \eqref{eq:intPsiEbeta} with $E_\beta$ as in \eqref{eq:Ebeta}.
\end{theorem}

In order to prove this theorem we need some auxiliary results.  
We first prove that  equation \eqref{eq:s9momeq} can be solved iteratively for $\| A\|$ small enough.  
\begin{lemma} \label{lem:lem1s9}
Let be $Q_{2} \in H_{2}(\C)$ such that 
 \be\label{eq:s9momeq2}
Q_{2}(k)-\mathcal{L}Q_{2}(k)+ (A_{\beta}k)\cdot\partial_{k}Q_{2}(k)=0.
\ee
Then, there exists an $ \tilde{\ep}_{M}=\tilde{\ep}_{M}(M,d)>0$ such that if $\|A\| \leq \tilde{\ep}_{M}$ 
there exists a unique family of homogeneous polynomials $\{Q_{\ell}(k)\}_{\ell=3}^{M}$ with $Q_{\ell} \in H_{\ell}(\C)$
for  $\ell= 3,\dots,M$ which satisfies \eqref{eq:s9momeq}, \eqref{eq:s9opK}.  Moreover, there exists a constant $C_M>0$ depending only on $d$ and $M$ such that 
\be \label{eq:bdQell}
\sup_{\ell=2,\dots, M} \|Q_{\ell} \| \leq C_M.
\ee 
\end{lemma}

\begin{proof}
We solve iteratively the system of equations \eqref{eq:s9momeq}. In order to do this we will prove that the operators $(I- \mathcal{L}_{\ell}+(A_{\beta}k)\cdot\partial_{k})$ are invertible in $H_{\ell}(\C)$ for $\ell=3, \dots, M$, if $\|A\|$ is sufficiently small. We first observe that the operators $(I-\mathcal{L}_{\ell})$ restricted to the space $H_{\ell}(\C)$ are invertible for any $\ell=3,\dots, M$.    
In order to prove this we consider the semigroups $e^{-t \mathcal{L}_{\ell}}:H_{\ell}(\C)\rightarrow H_{\ell}(\C)$. We have
\bes
\vert| e^{-t \mathcal{L}_{\ell}} \vphi \vert|_{\ell}= \sup_{k\in\R^d}\frac{\vert e^{-t \mathcal{L}_{\ell}} \vphi(k)\vert}{|k|^{\ell}}\leq \|\vphi\|_{\ell}  \sup_{k\in\R^d}\frac{\vert e^{-t \mathcal{L}_{\ell}} \vert \vert k\vert^{\ell}}{|k|^{\ell}}\leq  \|\vphi\|_{\ell} \; e^{(\lambda(\ell)-1)t}
\ees
where in the last inequality we used \eqref{eq:Lkp}. Hence the operator norm is estimated by
\be \label{eq:s9opnorm}
\vert| e^{-t \mathcal{L}_{\ell}}\vert|\leq e^{(\lambda(\ell)-1)t} \quad\text{for}\;\;  t\geq 0 
\ee
Using then that $\lambda(\ell)<0$ for $\ell\geq 3$ we obtain that $(\mathcal{L}_{\ell}-I)^{-1}$ exists in $H_{\ell}(\C)$ for any $\ell=3,\dots, M$ and it is given by the representation formula
\bes
(\mathcal{L}_{\ell}-I)^{-1}=\int_{0}^{\infty} dt \ e^{t (I-\mathcal{L}_{\ell})},\quad \ell=3,\dots, M.
\ees
Moreover, the inequality \eqref{eq:s9opnorm} implies
\be\label{eq:s9inv1}
\vert (\mathcal{L}_{\ell}-I)^{-1} \vert \leq \frac{1}{\vert \lambda(\ell)\vert },\quad \ell=3,\dots, M.
\ee
Notice that $(A_{\beta}k)\cdot\partial_{k}$ is an operator from $H_{\ell}(\C)$ to itself whose norm can be estimated by $ C_{\ell} \|A\|$. Thus we can use the Neumann series to show that $(I- \mathcal{L}_{\ell}+(A_{\beta}k)\cdot\partial_{k})$ is invertible in $H_{\ell}(\C)$ for  $3\leq \ell\leq M$ if $\|A\|\leq \tilde{\ep}_M$ where $\tilde{\ep}_M=\tilde{\ep}_M(d,M)$. Moreover, choosing $\tilde{\ep}_M$ sufficiently small we obtain 
\be \label{eq:s9inv2}
\vert|(I- \mathcal{L}_{\ell}+(A_{\beta}k)\cdot\partial_{k})^{-1}\vert|\leq \frac 2 {\vert \lambda(\ell)\vert },\quad \ell=3,\dots, M.
\ee
We can now solve iteratively equations \eqref{eq:s9momeq}, \eqref
{eq:s9opK}. We consider the case $\ell\geq 3$.  
We first notice that $\mathbb{K}_{\ell}\in H_{\ell}(\C)$. Then, using the invertibility of $(I- \mathcal{L}_{\ell}+(A_{\beta}k)\cdot\partial_{k})$ in $H_{\ell}(\C)$ we can write 
\be \label{eq:s9momeql}
Q_{\ell}=(I- \mathcal{L}_{\ell}+(A_{\beta}k)\cdot\partial_{k})^{-1} \left[\mathbb{K}_{\ell}\right],\quad \ell=3,\dots, M,
\ee 
whence the Lemma follows.
\end{proof}

We now prove that we can associate the moments $Q_{\ell}$ for $\ell=3,\dots, M$ to a characteristic function $\vphi\in\Phi$.  More precisely we have the following.

\begin{lemma}\label{lem:lem2s9}
Let be  $\|A\| \leq \ep_0$ with $\ep_0>0$  as in Theorem \ref{th:main1}. 
Suppose that $Q_2(k)=-\frac 1 2 N : k\otimes k$,  with $N \in Symm_{d\times d}(\R)$ as in item $(i)$ of  Theorem \ref{th:main1}. Let $\tilde{\ep}_M$ be as in Lemma \ref{lem:lem1s9}, $\|A\|\leq \tilde{\ep}_M$ and $\{Q_{\ell}(k)\}_{\ell=3}^{M}$, $Q_{\ell}\in H_{\ell}(\C)$ be the unique family of polynomials defined by means of \eqref{eq:s9momeq}, \eqref{eq:s9opK}. Then there exists ${\ep}_M\in (0,\tilde{\ep}_M)$ depending on $d, M$ such that if $\| A\| \leq \ep_M$ there exists $\varphi\in \Phi$ such that \eqref{eq:s9ffi} holds.
\end{lemma}
\begin{proof} 
We will obtain $\varphi\in \Phi$  as a perturbation of a characteristic function associated to a Gaussian probability measure. More precisely, 
we define $\vphi_g(k)=e^{-\vert k\vert^2}$. We have that $\vphi_g \in\Phi$ since it is the characteristic function associated to a Gaussian probability measure $f_g\in\mathcal{M}_{+}(\R^d)$ (actually a Maxwellian distribution). 

The higher order moments associated to $f_{g}$ can be computed in terms of the homogeneous polynomials $\bar{Q}_{\ell}(k)$ which are obtained by means of the following polynomial identity
\be\label{eq:defQbar}
 \bar{Q}_{2 j}(k)= (-1)^{j} \frac{\vert k\vert ^{2 j}}{j !},\quad  \bar{Q}_{2j+1}(k)=0 \quad \text{for}\quad j\geq 1.
\ee  
Notice that the polynomials $\bar{Q}_{\ell}(k)$ uniquely determine the moments (cf.~\eqref{eq:series})
\be\label{def:mombarm}
\bar{m}_{\alpha}=\int_{\R^d} dv \ v^{\alpha}  f_{g}(v).
\ee
We remark for further reference that for any probability distribution $f\in\mathcal{P}_{+}(\R^d)$ (not necessarily Gaussian) having moments $m_{\alpha}$ we have
\be\label{eq:s9Qm}
Q_{\ell}(k)=(-i)^{\ell} \sum_{|\alpha|=\ell} \frac{m_{\alpha} k^{\alpha}}{\alpha!},\quad \ell=0,1,2,\dots
\ee
Moreover, the polynomials $Q_{\ell}(k)$ determine uniquely the moments $m_{\alpha}$.  

We now construct the family of Hermite polynomials $\{H_{\alpha}(v)\}_{\alpha\in\N^d}$ associated to the Gaussian probability measure $f_g$. These polynomials satisfy the orthogonality condition (see for instance \cite{Bog})
\bes
\int_{\R^d_{+}} dv \ H_{\alpha}(v)H_{\beta}(v)  f_{g}(v)=\delta_{\alpha,\beta},\qquad \alpha, \beta\in\N^d.
\ees
Moreover, the polynomials $H_{\alpha}(v)$ have degree $\vert\alpha\vert$. 
Furthermore, the polynomials $H_{\alpha}(v)$ have the form 
\be \label{eq:formHalpha}
H_{\alpha}(v)=\sum_{|\beta|\leq |\alpha|}c_{\alpha,\beta}v^{\beta},\, \quad c_{\alpha,\beta}\in\R.
\ee
The values of the coefficients $c_{\alpha,\beta}$ are well-known but we will not need the explicit form in the rest of the argument. 

Notice that if $f\in \mathcal{P}(\R^d)$ has moments $m_{\alpha}$ with $|\alpha|\leq M$ we have 
\be\label{def:Palpha}
 \int_{\R^d}dv \ f(v)H_{\alpha}(v)=\sum_{|\beta|\leq |\alpha|}c_{\alpha,\beta}m_{\beta}:=P_{\alpha}(\{m_{\beta}\}_{|\beta|\leq |\alpha|}),\,\quad |\alpha|\leq M.
 \ee 
 We assume in all the following that $m_0=1$ and $m_{\alpha}=0$ for $|\alpha|=1$. 
We define the truncated polynomials $\{\tilde{H}_{\alpha,R}(v)\}$ as
\be
\tilde{H}_{\alpha,R}(v)= H_{\alpha}(v) \quad \text{for}\; |v|\leq R \quad \text{and}\quad \tilde{H}_{\alpha,R}(v)=0 \quad  \text{for}\; |v|> R,
\ee
where $R>0$ is chosen sufficiently large so that the matrix 
\be\label{eq:s9matrix}
\left( J_{\alpha,\beta;R} \right)_{0\leq |\alpha|,|\beta|\leq M}:=\left(\int_{\R^d_{+}} dv \ H_{\alpha}(v)\tilde{H}_{\beta,R}(v)  f_{g}(v)\right)_{0\leq |\alpha|,|\beta|\leq M}
\ee
is invertible.   This is possible due to the fact that $ J_{\alpha,\beta;R}\rightarrow \delta_{\alpha,\beta}$ as $R\to \infty$.
 
We now look for a probability density $f$ with the form
\be\label{eq:probf}
f(v)= f_{g}(v) \left(1+ \sum_{ 0\leq |\beta| \leq M} \xi_{\beta }\tilde{H}_{\beta,R}(v)\right),\quad \xi_{\beta }\in\R
\ee 

Given the polynomials $Q_{\ell}\in H_{\ell}(\C)$ satisfying \eqref{eq:s9momeq}, \eqref{eq:s9opK} for $\ell=3,\dots, M$ we choose $m_{\alpha}$ as in \eqref{eq:s9Qm}. Then, if the probability distribution $f$ exists, it would satisfy \eqref{def:Palpha}.  Assuming that $f$ has the form \eqref{eq:probf}  it would follow that
\bes
  \sum_{ 0\leq |\beta| \leq M} \xi_{\beta }J_{\alpha,\beta;R} =P_{\alpha}(\{m_{\beta}\}_{|\beta|\leq |\alpha|})-\int_{\R^d}dv \ f_g(v)H_{\alpha}(v) 
\ees
whence, thanks to the definition of $\bar{m}_{\alpha}$ (cf.~\eqref{def:mombarm}) and to \eqref{def:Palpha}, we obtain
\be \label{def:xibeta}
\sum_{ 0\leq |\beta| \leq M} \xi_{\beta }J_{\alpha,\beta;R}  =P_{\alpha}(\{m_{\beta}\}_{|\beta|\leq |\alpha|})- P_{\alpha}(\{\bar{m}_{\beta}\}_{|\beta|\leq |\alpha|}).
\ee
Using the fact that the matrix $\left( J_{\alpha,\beta;R} \right)_{0\leq |\alpha|,|\beta|\leq M}$ is invertible it follows that 
\bes 
\sup_{0\leq |\beta| \leq M} \vert  \xi_{\beta }\vert \leq C_{M}\sup_{0\leq |\beta| \leq M} \vert m_{\beta}-\bar{m}_{\beta}.\vert 
\ees
and using Lemma \ref{lem:hompol} we have
\bes 
\sup_{0\leq |\beta| \leq M} \vert m_{\beta}-\bar{m}_{\beta}\vert \leq C_{M} \sup_{\ell=2,\dots M}\| Q_{\ell}-\bar{Q}_{\ell} \|_{\ell}. 
\ees
whence 
\be \label{eq:bdxibeta} 
\sup_{0\leq |\beta| \leq M} \vert  \xi_{\beta }\vert \leq C_{M} \sup_{\ell=2,\dots M}\| Q_{\ell}-\bar{Q}_{\ell} \|_{\ell}.
\ee
We remark that when $A=0$ the solutions to \eqref{eq:s9CauchySt} are Gaussian due to Proposition \ref{prop:mom}. In particular one solution of  \eqref{eq:s9CauchySt} is given by 
$$\varphi(k)= e^{-\vert k \vert ^2}=\sum_{\ell=0}^{\infty }\bar{Q}_{\ell}(k)$$
with $\bar{Q}_{\ell}(k)$ given as in \eqref{eq:defQbar}. Plugging the Taylor expansion into \eqref{eq:s9CauchySt} as in the derivation of \eqref{eq:s9momeq} we obtain that the polynomials $\bar{Q}_{\ell}$ satisfy 
\be\label{eq:evbarQ}
\bar{Q}_{\ell}(k)-\mathcal{L}_{\ell}\bar{Q}_{\ell}(k)= \bar{\mathbb{K}}_{\ell}, \;\quad \; \bar{\mathbb{K}}_{\ell}=\mathbb{K}_{\ell}\left[\{\bar{Q}_{j}(k)\}_{j=1}^{\ell-1}\right],\quad \ell \geq 3
\ee
where $\mathbb{K}_{\ell}\big[\cdot\big]$ is as in \eqref{eq:s9opK}. 

We now look at the difference $\bar{Q}_{\ell}(k)- Q_{\ell}(k)$. We obtain
\be \label{eq:diffQell}
\left(\bar{Q}_{\ell}(k)- Q_{\ell}(k)\right) -\mathcal{L}_{\ell}\left(\bar{Q}_{\ell}(k)- Q_{\ell}(k)\right)= \bar{\mathbb{K}}_{\ell}-\mathbb{K}_{\ell}+ (A_{\beta} k)\cdot \partial_{k} Q_{\ell}(k),\qquad \ell \geq 3. 
\ee
Using our choice of $Q_2(k)$ in the lemma  as well as item $(ii)$ in Lemma \ref{lem:EPbeta0} 
we get 
\be \label{eq:bdQ2diff}
\| \bar{Q}_{2}(k)- Q_{2}(k)\|_{2} \leq C\| A\| .
\ee 
Using \eqref{eq:s9opK} and \eqref{eq:bdQell}, we obtain 
\bes 
\|\bar{\mathbb{K}}_{\ell}-\mathbb{K}_{\ell}\|_{\ell} \leq \tilde{C}_{M} \sup_{j=2,\dots,\ell-1} \| \bar{Q}_{j}(k)- Q_{j}(k)\|_{j}. 
\ees 
Thanks to the invertibility of the operator $(I-\mathcal{L}_{\ell})$ for $\ell\geq 3$ as well as $\vert|(A_{\beta}k)\cdot\partial_{k})^{-1}\vert|\leq C \vert|A_{\beta}\vert|$, where $\|\cdot\|$ is the operator norm from $H_{\ell}(\C)$ to $H_{\ell}(\C)$, and applying again estimate \eqref{eq:bdQell} we arrive at
\be
\| \bar{Q}_{\ell}(k)- Q_{\ell}(k)\|_{\ell} \leq C_{M} \left(\| A\| + \sup_{j=2,\dots,\ell-1} \| \bar{Q}_{j}(k)- Q_{j}(k)\|_{j}\right),\quad \ell \geq 3
\ee
Iterating this formula, using also \eqref{eq:bdQ2diff} as first step it follows that 
\bes 
\sup_{\ell=2,\dots,M}  \| \bar{Q}_{\ell}(k)- Q_{\ell}(k)\|_{\ell}\leq C_{M} \| A\|.
\ees 
Combining this estimate with \eqref{eq:bdxibeta} we obtain that, defining the coefficients $ \xi_{\beta }$ as in \eqref{def:xibeta}, they satisfy 
\bes
\sup_{0\leq |\beta| \leq M} \vert  \xi_{\beta }\vert \leq C_{M} \| A\|.
\ees 
Therefore if $\|A\|$ is sufficiently small we obtain that $ \left(1+ \sum_{ 0\leq |\beta| \leq M} \xi_{\beta }\tilde{H}_{\beta,R}(v)\right)$ in \eqref{eq:probf} is nonnegative whence $f$ defined by means of \eqref{eq:probf}  is a probability density satisfying \eqref{def:Palpha}. The fact that $f$ is a probability density follows since $m_0=1$. Using the fact the Hermite polynomials with degree $|\alpha|\leq M$ are a basis of the space of polynomials with degree $|\alpha|\leq M$ it then follows, from the definition of the functions $P_{\alpha}$ in \eqref{def:Palpha}, that the moments associated to the probability density $f$ are given by $m_{\alpha}$.  Hence $\vphi\in\Phi$.

\end{proof}

\begin{proofof}[Proof of Theorem \ref{th:higmom}]
The fact that the set $\mathcal{F}_{p,M}(N)$ is nonempty  follows from Lemma \ref{lem:lem2s9}.  
The existence and uniqueness of the stationary solution $\Psi(\cdot)\in \mathcal{F}_{p,M}(N)$ can be proved by means of a fixed point argument in $\mathcal{F}_{p,M}(N)$ as in the proof of Theorem \ref{th:main1}. 
Notice that the equations for the higher order moments $Q_{\ell}(k)$ (cf. \eqref{eq:s9momeq}-\eqref{eq:s9opK}) have to be obtained using the corresponding mild formulation for the equation, in terms of the semigroup $E_\beta(t)$, as it was made  in the proof of Theorem \ref{th:main1} for the second order moments $Q_2(k)$. Here we omit the details for the sake of brevity.
\end{proofof}

\begin{remark}
We observe that Theorem \ref{th:higmom} combined with the uniqueness result in Theorem  \ref{th:main1} imply that the solution to \eqref{eq:s9CauchySt} has finite moments of any order if $\|A\|$ is sufficiently small. In our proof we use strong smallness conditions to obtain finite higher order moments. It would be interesting to clarify if we can have infinitely many finite moments if $\|A\|$ is small. This would require more refined estimates than the ones in this paper.
\end{remark}

\section{Reformulation of the results in terms of the probability distribution $f$} \label{sec:10}

In this section we reformulate the two main results of the paper, i.e. Theorem \ref{th:main1} and Theorem \ref{th:main2} in terms of the original measure $f\in\mathcal{P_{+}}(\R^d)$.
\smallskip

We first define the function $f_{st}(\cdot) \in \mathcal{P}_{+}(\R^d)$ by means of
\be\label{eq:deffst}
f_{st}(v)=\frac  1 {(2\pi)^{d}} \int_{\R^d} dk \  e^{i k\cdot v } \Psi(k)
\ee
where $\Psi\in \mathcal{F}_p(N)$ is the solution to \eqref{eq:intPsiEbeta} obtained in Theorem \ref{th:main1}.  Notice that formula \eqref{eq:deffst} has to be understood as an identity in the sense of  distributions, i.e. acting on a  Schwartz test function (cf. \cite{RS2}). 

Theorem \ref{th:main1} implies the existence of a particular class of self-similar solutions to \eqref{eq:GenBolt}, \eqref{eq:CollBolt}. More precisely,
\begin{theorem}\label{th:exssf}
Let $2<p\leq 4$. There exists $\ep_0>0$ depending on $p$ and $d$ such that if $\|A\|\leq \ep_0$ the following formula provides an explicit self-similar solution of \eqref{eq:GenBolt}, \eqref{eq:CollBolt} :
\be
f(v,t)= \big( e^{-\beta t} \big)^d f_{st}\left(\frac{v- e^{-tA^T} U}{e^{\beta t}}\right), \quad U\in\R^d,
\ee
such that $\int_{\R^d} dv \ f_{st}(v)|v|^p<\infty$.  Here $f_{st}$ is given as in \eqref{eq:deffst}. Moreover, for every $M\in\N$ there exists an $\ep_M\in (0,\ep_0)$ sufficiently small such that if $\| A\|\leq \ep_M $  then $\int_{\R^d} dv \ f_{st}(v)|v|^M<\infty$.
\end{theorem}

\begin{proof}
This result is a consequence of Theorems \ref{th:main1}, \ref{th:higmom} as well as \eqref{eq:PssGen} and  \eqref{eq:chvarSSPtoSS}.
\end{proof}

We now reformulate in the space of probability measures the stability result obtained for characteristic functions (cf.~Theorem \ref{th:main2}).

\begin{theorem}\label{thm:Main_f}
Suppose that $f\in C([0,\infty);\mathcal{M}_{+}(\R^d))$ is a solution to \eqref{eq:GenBolt}, \eqref{eq:CollBolt} with initial datum  $f(\cdot,0)=f_0(\cdot)\in \mathcal{P}_{+}(\R^d)$ such that $\int_{\R^d} dv\ f_0(v) |v|^p<\infty$ for $p>2$. Moreover, we define 
\be
U= \int_{\R^d} dv\ f_0(v) v\quad \text{and}\quad  b_{j,\ell}=2 \int_{\R^d} dv\ f_0(v) (v-U)_j (v-U)_{\ell}, \quad B_0=(b_{j,\ell})_{{j,\ell}=1}^d.
\ee
Suppose that $\|A\|\leq \ep_0$ with $\ep_0$ as in Theorem \ref{th:main2} and $\beta\in\R$ is the eigenvalue with largest real part of \eqref{eq:EVpbMom}. Let $\lambda$  be as in \eqref{eq:asymB} in Lemma \ref{lem:asym} where $B(t)$ is the solution of \eqref{eq:SecMomEq2} with initial datum $B(0)=B_0.$
Then, 
\be
\left( e^{\beta t} \right)^d \ f \left(e^{\beta t}v+e^{-tA^T} U,t\right) \rightarrow  \lambda^{-d}f_{st}\big({\lambda}^{-1} v\big) \quad \text{as}\quad t\to\infty 
\ee
in the weak topology of $\mathcal{P}_{+}(\R^d)$ and where $f_{st}$ is as in \eqref{eq:deffst}. 
\end{theorem}

\begin{remark}
We notice that in the case $\lambda=0$ we have $f_{0}\left(  v\right)  =\delta\left(  v-U\right)$ and the solution obtained in Theorem \ref{th:exssf}, and Theorem \ref{thm:Main_f} is just a shifted Dirac Delta function which corresponds to the solution $\vphi(k,t) =1$ discussed in Remark \ref{rk:dirac}.
\end{remark}

\begin{proof}
At the level of the characteristic functions $\vphi(\cdot,t)$, the proof is a consequence of Theorem \ref{th:main2} combined with the computations made in Section \ref{sec:6} to see which is the effect of the original average velocity $U$ (cf. \eqref{eq:PssGen}), as well as the change of variables \eqref{eq:chvarSSPtoSS} which was made to transform self-similar solutions to the steady state. 
In order to obtain the result for the probability distribution $f(\cdot,t)$ we use elementary properties of the Fourier transform that relates the group of translations with the multiplication by a phase, the scaling properties of the Fourier transform as well as the fact that the uniform convergence of characteristic functions in compact sets implies the weak convergence of the original probability distributions (cf.~\cite{Fe2}). 
\end{proof}

\section{Conclusions}\label{sec:11}

We have considered in detail the initial value problem \eqref{eq:GenBolt} for the Maxwell-Boltzmann equation with collision term \eqref{eq:CollBolt} and studied the long time asymptotics of the solutions. Notice that, in addition to the classical collision operator, \eqref{eq:GenBolt} contains a linear stretching term with the form $\operatorname{div}\left(  Avf\right).$  Applications of this problem are  connected with the well-known class of homoenergetic solutions, which are a particular class 
of solutions to the spatially inhomogeneous Boltzmann equation. They are characterized by having the
same dispersion of velocities at all the points $x\in\mathbb{R}^{3}$ for each
given time $t\in\mathbb{R}$. These solutions have been studied by several authors since 1950s.  Our main goal was to study a large time asymptotics of  solutions under the assumption of smallness of the matrix $A$.  The investigation was performed in the Fourier representation on the basis of  some special properties (see, in paricular, Section \ref{sec:3}) of  the non-linear collision operator in such representation.  The main  result of  the paper is formulated in Theorem \ref{thm} (Section \ref{sec:2}). Roughly speaking, it says that,  for sufficiently small norm of $A$,  any non-negative solution with finite second moment tends to a self-similar solution of relatively simple form  for large values of time. This is what we call ``the self-similar asymptotics".  A similar property is known to be valid for quasi-isotropic or one-dimensional  Maxwell models \cite{BCG09}, but in our case the self-similar solutions are strongly anisotropic.  To the best of our knowledge there are no previous  results on the convergence to anisotropic self-similar solutions. We also gave an alternative proof to the one in \cite{JNV} of existence of  the self-similar solutions and obtained some of their properties, in particular, the existence of  moments of any given order for sufficiently small norm of $A$ (depending on that maximal order).

 We finish the paper with  some open  questions. Our approach can be relatively easily extended to the case of moderate values of the norm of $A$ in \eqref{eq:GenBolt}. It is still not clear if all moments of the self-similar solution  are finite even for small $A$.  The case of large $A$ is more difficult except for exactly solvable (for self-similar solutions) version of \eqref{eq:GenBolt}, where $A$ is proportional to the identity matrix. It is interesting to study the high energy tail and finiteness of moments for particular case of shear flow with nilpotent matrix $A$ s.t. $A^2 = 0$.  There are some formal arguments by physicists (see e.g. \cite{GS}) suggesting that the fourth moment of  the self-similar profile becomes infinite for large values of  $A$ (strong shear), but there are no rigorous results on that problem.  Hopefully the methods of this paper can be useful for solving these and similar problems in the future.


\section*{Appendix A: Derivation of the matrix equation for $B(t)$} \label{app:A}

To obtain the equation for the second moments (cf.\eqref{eq:SecMomEq}) we compute first the left hand side  of \eqref{eq:BoltFourier2} acting over functions $\vphi$ with the form \eqref{eq:expffi_f2}. Then it becomes
\begin{align}\label{eq:lhsMom}
&-\frac 1 2 \sum_{j,\ell=1}^{d} \partial_t b_{j,\ell}(t) \ k_j k_{\ell} -\frac 1 2\sum_{j}^{d} \sum_{r,s=1}^{d}  \left(A_{r,j}b_{r,s}(t) \ k_j k_{s}+A_{s,j}b_{r,s}(t) \ k_j k_{r} \right) \nonumber \\&
=-\frac 1 2 \sum_{j,\ell=1}^{d} \partial_t b_{j,\ell}(t) \ k_j k_{\ell} -\frac 1 2\sum_{j}^{d} \sum_{r}^{d}\sum_{\ell=1}^{d}  A_{r,j}b_{r,\ell}(t) \ k_j k_{\ell}
\end{align}  
while the right hand side of \eqref{eq:BoltFourier2} gives
\begin{align}\label{eq:rhsMom2}
 \Gamma[\vphi](k)-\vphi(k)& =
-\frac 1 4 \int_{S^{d-1}} dn \ g(\hat{k}\cdot n) \left[ \sum_{j,\ell=1}^{d}b_{j,\ell}(t) \left( k_{j}k_{\ell}+|k|^2 n_{j}n_{\ell}\right)\right] +\frac 1 2\sum_{j,\ell=1}^{d} b_{j,\ell}(t) \ k_j k_{\ell} + O(\vert k\vert^p)\nonumber \\&
=\frac 1 2 \sum_{j,\ell=1}^{d}b_{j,\ell}(t) \int_{S^{d-1}} dn   \ g(\hat{k}\cdot n)  \left[ \frac 1 2 k_{j}k_{\ell} -\frac 1 2 \vert k \vert^2 n_j n_\ell\right]+O(\vert k\vert^p)\nonumber \\&
=\frac 1 4 \sum_{j,\ell=1}^{d}b_{j,\ell}(t)  T_{j,\ell}+O(\vert k\vert^p)
\end{align}
where 
\be
 T_{j,\ell}:= \int_{S^{d-1}} dn   \ g(\hat{k}\cdot n)    \big(  k_{j}k_{\ell} - \vert k \vert^2 n_j n_\ell \big).
\ee
Notice that $T_{j,\ell}$ is a traceless isotropic tensor depending on $k$. Therefore, it has the form  $ T_{j,\ell}= C_0\big(k_{j}k_{\ell} -\frac {\vert k \vert^2}{ d} \delta_
{j,\ell}\big)$ for some $C_0\in\R$. 

Computing the tensor for the particular choice of the vector $k=(1,0,\dots)$ we obtain
\be \label{eq:defq2}
C_0\left(\frac{d-1}{d}\right)= \int_{S^{d-1}} dn   \ g(\hat{k}\cdot n)    \big(  1- (\hat{k}\cdot n)^2  \big) :=q
\ee
whence
\bes
 T_{j,\ell}= \left(\frac{d}{d-1}\right) q \big(  k_{j}k_{\ell} - \frac 1 d \vert k \vert^2 \delta_{j,\ell} \big). 
\ees 
Thus \eqref{eq:rhsMom2} becomes
\be\label{eq:rhsMom}
\Gamma[\vphi](k)-\vphi(k)= \frac{q d}{4(d-1)}  \sum_{j,\ell=1}^{d}b_{j,\ell}(t) \big(  k_{j}k_{\ell} - \frac 1 d \vert k \vert^2 \delta_{j,\ell} \big)+O(\vert k\vert^p).
\ee

Thus, combining \eqref{eq:lhsMom} with \eqref{eq:rhsMom} and symmetrizing in the second sum of \eqref{eq:lhsMom}, we obtain 
\begin{align*}
& \partial_t b_{j,\ell}(t) + \sum_{r=1}^{d} \left( A_{r,j}b_{r,\ell}(t)+A_{r,\ell}b_{r,j}\right) =-  \frac{q d}{2(d-1)} b_{j,\ell}(t) + \frac{q}{2(d-1)} \Tr(B) \delta_{j,\ell} 
\end{align*}
that can be rewritten in matrix form as 
\be
\partial_t B+ \big(BA +(BA)^T\big)+\frac{q d}{2(d-1)}  \left( B - \frac{\Tr(B)}{d}I\right)=0
\ee 
where $M^T$ denotes the transposed of the matrix $M$.

\bigskip

\textbf{Acknowledgements.}
The authors gratefully acknowledge the support of the Hausdorff Research Institute for Mathematics
(Bonn), through the {\it Junior Trimester Program on Kinetic Theory},
and of the CRC 1060 {\it The mathematics of emergent effects} at the University of Bonn funded through the German Science Foundation (DFG). A.V.B. also acknowledges the support of Russian Basic Research Foundation (grant 17 - 51 - 52007).


\bigskip

\bigskip

\adresse

\end{document}